\documentclass[11pt]{article}

\usepackage{fullpage}
\usepackage{graphicx,enumerate}
\usepackage{float}

\newcommand{\p}{\mathbf{p}}
\newcommand{\pp}{\mathbf{P}}
\newcommand{\sw}{SW}
\newcommand{\opt}{OPT}
\newcommand{\s}{\mathbf{s}}
\newcommand{\poa}{\text{\normalfont PoA}}
\newcommand{\pos}{\text{\normalfont PoS}}
\newcommand{\ds}{\mathcal{D}}
\newcommand{\remove}[1]{}

\newcommand{\mk}[1]{{\color{black}{#1}}}
\newcommand{\tc}[1]{{\color{black}{#1}}}
\newenvironment{proof-sketch}{\noindent{\it Sketch of Proof.}\hspace*{1em}}{\qed\bigskip}
\newenvironment{proof-folklore}{\noindent{\it Proof (folklore).}\hspace*{1em}}{\qed\bigskip}

\usepackage[linesnumbered,boxed,vlined,ruled]{algorithm2e}
\usepackage{caption}
\usepackage{subfigure}
\usepackage{graphicx}
\usepackage{amsmath}
\usepackage{amsmath}
\usepackage{amssymb}
\usepackage{amsthm}
\usepackage{appendix}
\usepackage{todonotes}
\usepackage{url}

\usepackage{multirow}

\newtheorem{theorem}{Theorem}
\newtheorem{lemma}[theorem]{Lemma}

\newtheorem{corollary}[theorem]{Corollary}

\theoremstyle{definition}
\newtheorem{example}{Example}
\newtheorem{remark}{Remark}

\begin{document}
\title{\bf Financial network games}

\author{Panagiotis Kanellopoulos \and
Maria Kyropoulou \and
Hao Zhou}

\date{School of Computer Science and Electronic Engineering \\ University of Essex, UK}

\maketitle
\begin{abstract}
We study financial systems from a game-theoretic standpoint. A financial system is represented by a network, where nodes correspond to firms, and directed labeled edges correspond to debt contracts between them.
The existence of cycles in the network indicates that a payment of a firm to one of its lenders might result to some incoming payment. So, if a firm cannot fully repay its debt, then the exact (partial) payments it makes to each of its creditors can affect the cash inflow back to itself.  We naturally assume that the firms are interested in their financial well-being (\emph{utility}) which is aligned with the amount of incoming payments they receive from the network.  This defines a game among the firms, that can be seen as utility-maximizing agents who can strategize over their payments.

We are the first to study \emph{financial network games} that arise under a natural set of payment strategies called priority-proportional payments. We investigate the existence and (in)efficiency of equilibrium strategies, under different assumptions on how the firms' utility is defined, on the types of debt contracts allowed between the firms, and on the presence of other financial features that commonly arise in practice. Surprisingly, even if all firms' strategies are fixed, the existence of a unique payment profile is not guaranteed. So, we also investigate the existence and computation of valid payment profiles for fixed payment strategies.

\end{abstract}
\section{Introduction}\label{sec:introduction}
A financial system comprises a set of institutions, such as banks and firms, that engage in financial transactions. The interconnections showing the \emph{liabilities} (financial obligations or debts) among the firms are represented by a network and can be highly complex. Hence, the wealth and financial well-being of a firm depends on the entire network and not just on the well-being of its immediate borrowers/\emph{debtors}.
For example, a possible bankruptcy of a firm and the corresponding damage to its immediate lenders/\emph{creditors} resulting from the firm's failure to repay (known as credit risk), can be propagated through the financial network by causing the creditors' (and other firms', in sequence) inability to repay their debts, thus having a global effect.%; even lead to a major financial crisis.
%Understanding how the lack of necessary funds of some firms is spread through the network, and the effect this has on the set of financially healthy firms (which are able to fulfil their obligations), is an important objective towards the main goal of managing the systemic risk of instability or collapse of the entire network/economy.

In this work, we examine the global effect of payment decisions of individual firms. We assume that each firm has a fixed amount of \emph{external assets} (not affected by the network) which are measured in the same currency as the liabilities. A firm's \emph{total assets} comprise its external assets and its incoming payments, and can be used for (outgoing) payments to its creditors. A firm's payment decision, for example, can 
specify the priority to be given to each of its debts or creditors.
%include prioritising the payment of certain creditors and/or deciding the percentage of its assets that each creditor will get up to the point that the corresponding debt will be paid for (note that the actual amount of incoming payments is not known in advance). 
An authority (such as the government or the court) is assumed to monitor the decisions of the firms in order to guarantee that they comply with general regulatory principles,  such as the absolute priority and the limited liability ones (see, e.g., \cite{Eisenberg01}). According to these principles, a firm can leave a liability unpaid only if its assets are not enough to fully repay its liabilities.
In particular, the \emph{absolute priority principle}  requires that all creditors must be paid off before a firm's stakeholders can split assets, and the \emph{limited liability principle} implies that a firm that does not have enough assets to pay its liabilities in full, has to spend all its remaining assets to pay its creditors. Surprisingly, not all decisions of the firms lead to valid solutions (or payment profiles)\footnote{An example of such a network appears in \cite{DefaultAmbiguity19} and assumes the presence of default costs and Credit Default Swaps (definitions appear in Section \ref{sec:preliminaries}).}, or there can be ambiguity in the overall solution despite the individual fixed payment decisions of the firms\footnote{Consider, for example, four firms A, B, C, and D, with zero external assets each, and the following liabilities: A owes $\ell$ coins to B and another $\ell$ coins to C; B owes $\ell$ coins to A and another $\ell$ coins to D. Assume that the payment strategy of firm A is to prioritize payments towards B over C, and that the payment strategy of firm B is to prioritize payments towards A over D (firms C and D are not strategic since they do not have multiple creditors---or any for that matter). Then there exist infinitely many valid solutions that are consistent with this strategy profile: Indeed, A paying B the amount of $\lambda$ coins and B paying A the same amount ($\lambda$ coins), is a valid solution for any positive value $\lambda\leq \ell$.} \cite{Eisenberg01}. 
This gives rise to the \emph{clearing problem} aiming to determine the final payment profile of the network for fixed individual payment decisions. Such \mk{payments are called \emph{clearing payments}} and are compatible with the given individual payment decisions.
Computing \mk{clearing payments} is a necessary step before performing a game-theoretic analysis.

We investigate the consequences of individual firm decisions in clearing from a game-theoretic standpoint \cite{nisan_roughgarden_tardos_vazirani_2007}.  
%In an attempt to understand the consequences of individual firm decisions in clearing, we view each firm as a strategic agent that can decide whom among their creditors to prioritize when paying their debt. Therefore, we investigate the possible outcomes from a game-theoretic standpoint \cite{nisan_roughgarden_tardos_vazirani_2007}.  
Following the recent work of Bertschinger et al. \cite{Hoefer19}, we diverge from the common assumption of proportionality in payments, that dictates that a firm pays its creditors proportionally to their respective liabilities.
%, i.e., if a firm's assets can only cover 20\% of its total liabilities, it will pay each of its creditors 20\% of the corresponding liability. 
We assume that firms strategize about their payments by appropriately \mk{deciding the priority of their payment actions}, consistently with a predefined payment scheme (and the current regulation)\mk{; the chosen priorities may have an effect on the incoming payments of a given firm}. We consider a natural payment scheme that allows  allocations of creditors to priority-classes independently of the available assets.
We refer to this payment scheme  as priority-proportional scheme, and the corresponding strategies as \emph{priority-proportional  strategies}.

Our paper is the first to perform a game-theoretic analysis of priority-proportional payments in financial networks. Payments with priorities are simple to express as well as quite common and very well-motivated in the financial world. Indeed, bankruptcy codes allow for assigning priorities to the payout to different creditors, in case an entity is not able to repay all its obligations. Such a distribution of payments can be part of a reorganization plan ordered by the court \cite{Priorities}.
%It is also possible that the distribution of payments is dictated by a reorganization plan ordered by the court \cite{Priorities}. 
Priority classes in bankruptcy law have also been considered in \cite{Elsinger09,Kaminski00}, among others. Moving beyond regulated financial contexts, similar behavior is common in everyday transactions between individuals with pairwise debt relations. This is the first time that the impact of priority-proportional payments is assessed in a strategic setting, as are other elements of our analysis, even though such payments have been considered in the past (most recently by Papp and Wattenhofer \cite{wattenhofer2020network}). Our work demonstrates that, despite their simplicity, priority-proportional strategies exhibit certain desired properties with respect e.g., to equilibria existence and quality. They are also not very restrictive, which can lead to unattractive instances.
%hence, allow for unattractive instances too.
%Priority-based payment schemes appear in financial networks, albeit not in a strategic setting. 
Our work can be seen as an attempt to examine whether such priority-based payment plans are indeed stable equilibria; this can be useful from a mechanism design perspective. 
%Compared to coin-ranking paymentsl, priority-proportional ones are less general but are simple to express and exhibit certain desired properties with respect e.g., to equilibria existence and quality. 

Our game-theoretic analysis considers two different definitions of utility motivated by the financial literature, namely \emph{total assets}, computed as the sum of external assets and incoming payments, or \emph{equity}, respectively. Traditionally, the financial health of a firm has been measured by its equity, which is equal to the amount of remaining assets after payments (total assets minus liabilities) if this is positive, and $0$ otherwise. This means that all firms that have more debt than assets have equity $0$, so equities fail to capture the potentially different available assets these firms might have. Total assets can be seen as a refinement of the equities: indeed, for any financially healthy firm (that can repay all its debt) its total assets equal its equity plus its liabilities (a fixed term), while for other firms,
total assets allow them to distinguish among different states that are indistinguishable in terms
of equity. Note that, since by the absolute priority principle a firm is obliged by law to fully repay all its creditors if it has enough assets, it is only the payment decisions of firms whose debt is more than their assets that can have an effect on the network. For this reason, computing the firms' utility as their total assets is a suitable approach in a game-theoretic
context. Overall, both definitions are aligned with the individual financial wealth and welfare of a firm, so maximizing the utility is a firm's natural individual objective. 

Our model is based on the seminal and widely adopted work of Eisenberg and Noe \cite{Eisenberg01} who provide a basic financial network model, i.e., proportional payments, non-negative external assets,  and debt-only liabilities.  In addition to considering a different payment scheme, we enhance  the basic model by also considering financial features commonly arising in practice, such as  default costs \cite{rogers2013failure}, Credit Default Swap (CDS) contracts \cite{DefaultAmbiguity19}, and negative external assets \cite{Dem18} (definitions appear in Section \ref{sec:preliminaries}).  
%\emph{Default costs} capture the fact that firms in default might lose part of the value of their assets (both external assets and their incoming payments) because of the urgent need to liquidate them in order to partially repay their liabilities.  \emph{CDS}s are contracts of conditional payments subject to the default of a reference entity. CDSs can play a role of insurance for default risks, when a firm holds a CDS which refers to one of its debtors. Negative external assets can correspond to debts a firm has towards creditors external to the network (formal definitions are provided later in the paper). 
We analyze the efficiency of the states arising from various \mk{clearing payments of financial networks, with a focus on ones consistent with equilibrium strategies}.  Note that even though, as stated above, both utility definitions are aligned with the individual financial wealth and welfare of a firm, it turns out that these notions of individual utility are not always aligned with the welfare of the whole financial system.
% (sum of individual utilities), and a small benefit for a single firm can come at the expense of a considerable global financial damage. 

 Overall, we aim to quantify the extent to which strategic behavior of the firms affects the welfare of the society, by analyzing the \emph{financial network games} that are defined by a particular utility function, and possibly allow the presence of 
%default costs, CDSs, and negative external assets is allowed. We identify Nash equilibria of such financial network games and argue about their (in)efficiency in terms of the total welfare of the system.
other financial instruments. In particular, we consider financial network games under priority-proportional strategies, defined for different utility functions, such as total assets or equities,  and which potentially allow CDS contracts, default costs, or negative external assets.  We derive structural results that have to do with the existence, the computation, and the properties of \mk{clearing payments for fixed payment decisions in a non-strategic setting, and/or the existence and quality of equilibrium strategies}. In particular, in Section \ref{sec:non-strategic} we prove the existence of  \mk{maximal clearing payments} under priority-proportional strategies, even in the presence of default costs, and provide an algorithm that computes \mk{them} efficiently. We are then able to prove existence of equilibria when the utility is defined as the equity, but show that equilibria are not guaranteed to exist when the utility is captured by the total assets. We then turn our attention to the efficiency of  equilibria and provide an almost complete picture of the price of anarchy \cite{Koutsoupias:2009} and the price of stability \cite{Schulz-pos03,Anshelevich-pos08}. Our results for total assets appear in Section \ref{sec:assets}, while the case of equities is treated in Section \ref{sec:equities}.

\paragraph{Related work}
Financial networks and their related properties have been analyzed in various works that follow the standard (non-strategic) model developed by Eisenberg and Noe \cite{Eisenberg01}. They introduce a financial network model allowing debt-only contracts among firms with non-negative external assets that make proportional payments. Among other results, they prove that there always exist maximum and minimum clearing \mk{ payments} and  identify sufficient conditions so that  uniqueness is guaranteed; they also present an efficient iterative algorithm that computes them.

Following the model of Eisenberg and Noe \cite{Eisenberg01}, a series of papers extend and enrich the model by adding default costs \cite{rogers2013failure}, cross-ownership \cite{elliott2014financial,vitali2011network}, liabilities with various maturities \cite{allouch2017strategic,kusnetsov2019interbank} and CDS contracts \cite{wattenhofer2020network,DefaultAmbiguity19}. In \cite{rogers2013failure}, Rogers and Veraart prove the existence of maximal clearing \mk{payments} in the presence of default costs and provide an algorithm that computes them. Schuldenzucker et al. \cite{DefaultAmbiguity19}  also consider CDS contracts and show that, in general, there can be zero or many clearing \mk{payments}. Furthermore, they provide sufficient conditions for the existence of unique clearing \mk{payments}. Demange \cite{Dem18} proposes a model that allows for firms to be indebted to entities outside the network, and captures this by allowing negative external assets.

The strategic aspects of financial networks have been considered very recently. Most relevant to our setting is the paper by Bertschinger et al. \cite{Hoefer19} who also study the inefficiency of equilibria in financial networks. They follow the standard model of \cite{Eisenberg01} and focus mainly on two payment schemes, namely coin-ranking and edge-ranking strategies, while also considering the total assets (as opposed to equity) as a measure of the individual utility of each firm. Apart from defining the graph-theoretic version of the clearing problem, they present a large range of results on the existence and quality of equilibria. Our work extends this line of research by considering a different payment scheme \mk{(which extends edge-ranking strategies by allowing ties in the ranking)} and the possibility of additional common financial features (default costs, CDS contracts, and negative external assets). \mk{In a similar spirit, Papp and Wattenhofer \cite{wattenhofer2020network} consider} the impact of individual firm actions, such as removing an incoming debt, donating extra funds to another firm, or investing more external assets, when CDS contracts are also allowed. They mostly focus on the case where firms have predefined priorities over their creditors and they remark that even by redefining such priorities, a firm cannot affect its equity.

Schuldenzucker and Seuken \cite{SS20} consider the problem of portfolio compression, where a set of liabilities forming a directed cycle in the financial network may be simultaneously removed, if all participating firms approve it. They consider questions related to the firms' incentives to participate in such a compression, while Schuldenzucker et al. \cite{SSB17} show that finding clearing payments when CDS contracts are allowed is PPAD-complete.  Additional strategic considerations, albeit less related to our setting, are the focus of Allouch and Jalloul  \cite{allouch2017strategic} who consider liabilities with two different maturity dates and study how firms may strategically deposit some amount from their first-period endowment in order to increase their assets in the second period, while Cs\'{o}ka and Herings \cite{CH19} consider liability games and study how to distribute the assets of a firm in default, among its creditors and the firm itself.

\section{Preliminaries}\label{sec:preliminaries}
In the following, we denote by $[n]$ the set of integers $\{1, \dots, n\}$. The critical notions of this section and their graphical representation are presented in an example (Figure \ref{fig:example}), at the end of this Section.

\paragraph{Financial networks.} Consider a set $\mathcal{N}=\{v_1, \dots, v_n\}$ of $n$ firms, where each firm $v_i$ initially has some \emph{external assets} $e_i$ corresponding to income received from entities outside the financial system;  note that $e_i$ may also be negative. 
%Let $\mathbf{e}=(e_1, \dots, e_n)$ be the vector of external assets.

Firms have payment obligations, i.e., \emph{liabilities}, among themselves. These are in the form of a simple debt contract or of a Credit Default Swap. A \emph{debt contract} creates a liability $l^0_{ij}$ of firm $v_i$ (the debtor) to firm $v_j$ (the creditor); we assume that $ l_{ij}^0 \geq 0$ and $l_{ii}^0=0$. Note that both $l_{ij}^0>0$ and $l_{ji}^0>0$ may hold. Firms with sufficient funds to pay their obligations in full are called \emph{solvent} firms, while ones that cannot are \emph{in default}. The \emph{recovery rate}\mk{, $r_k$, of a firm $v_k$ that is in default, is defined as} the fraction of its total liabilities that it can fulfill.  A \emph{Credit Default Swap} (CDS) is a conditional liability of $(1-r_k)l_{ij}^k$ of the debtor $v_i$ to the creditor $v_j$, subject to the default of $v_k$, called the \emph{reference entity}.  
Overall, the total liability of firm \mk{$v_i$ to firm $v_j$} is $l_{ij} = l_{ij}^0+\sum_{k \in [n]} {(1-r_k)l_{ij}^{k} }.$
%\begin{displaymath}
%l_{ij} = l_{ij}^0+\sum_{k \in [n]} {(1-r_k)l_{ij}^{k} }.
%\end{displaymath}
Let $L_i=\sum_j {l_{ij}}$ be the total liabilities of firm $v_i$ and set $\mathbf{l}=(L_1, \dots,L_n)$. 
%We represent all liabilities with a 3-dimensional matrix $\mathbf{L}=(l_{ij}^k)$ with $i, j \in [n]$ and $k\in \{0\}\cup [n]$.

Let $p_{ij}$ denote the payment from $v_i$ to $v_j$; we assume that $p_{ii} =0$. These payments define a payment matrix $\mathbf{P} = (p_{ij})$ with $i,j \in [n]$. 
Then, $p_i = \sum_{j \in [n]}{p_{ij}}$ represents the total outgoing payments of firm $v_i$, while $\mathbf{p} = (p_1, \dots, p_n)$ is the payment vector\mk{; this should not be confused with the breakdown of individual payments of firm $v_i$ that is denoted by $\p_i=(p_{i1}, \dots, p_{in})$}.
%, which with a slight abuse of notation we use as a proxy for $\mathbf{P}$. 
A firm in default may need to liquidate its external assets or make payments to entities outside the financial system (e.g., to pay wages). This is modeled using \emph{default costs} defined by values $\alpha, \beta \in [0,1]$. A firm in default \mk{can only} use an $\alpha$ fraction of its external assets (when this is positive) and a $\beta$ fraction of its incoming payments.  The absolute priority and limited liability regulatory principles, discussed in the introduction, imply that a solvent firm must repay all its obligations to all its creditors, while a firm in default must repay as much of its debt as possible, taking default costs also into account. Summarizing, it must hold that $p_{ij}\leq l_{ij}$ and, furthermore, 
$\mathbf{P}=\Phi(\mathbf{P})$, where
\begin{equation}\label{eq:fixed-point}
\Phi(\mathbf{x})_i= \left\{
\begin{array}{ll}
L_i, & \textrm{if } L_i\leq e_i+\sum_{j=1}^n x_{ji}\\
\alpha e_i+\beta \sum_{j=1}^n x_{ji}, & \textrm{otherwise}.\\
\end{array} \right.
\end{equation}

\mk{Payments $\mathbf{P}$ that satisfy these constraints are called \emph{clearing payments}}\footnote{\mk{Clearing payments are not necessarily unique.}}. \mk{We define the notion of \emph{proper} clearing payments, which are clearing payments that satisfy that all the money circulating in the financial network have originated from some firm with positive external assets. In the following, we only consider proper clearing payments.} 
%Overall, a financial network can be fully represented by a $4$-tuple $(\mathbf{L}, \mathbf{e}, \alpha, \beta)$. 

\paragraph{Financial network games.} These games arise naturally when we view the firms as strategic agents. We denote the strategy of firm $v_i$ by $s_i(\cdot)$; this dictates how $v_i$ allocates its existing funds, for any possible value these might have. 
%Given a payment vector $\mathbf{p}$ (and, implicitly, a payment matrix $\mathbf{P}$), we denote by $s_i(\mathbf{p})$ the strategy of $v_i$; in this case, $v_i$'s total funds are $e_i+\sum_{j \in[n]}{p_{ji}}$. 
Similarly, we can define the strategy profile $\mathbf{s} = (s_1(\cdot), \dots, s_n(\cdot))$. 

Given clearing payments $\mathbf{P}$, we define firm $v_i$'s utility using either of the following two notions. The \emph{total assets} $a_i(\mathbf{P})$ (see also \cite{Hoefer19,CH18}) are $$a_i(\mathbf{P})=e_i+\sum_{j\in [n]} p_{ji},$$
while the \emph{equity} $E_i(\mathbf{P})$ (see also \cite{wattenhofer2020network}) is $$ E_i(\mathbf{P})=\max\{0, a_i(\pp)-L_i\}.$$ 

%For example, a clearing payment vector for a financial network $(\mathbf{L}, \mathbf{e}, \alpha, \beta)$ under proportional payments is a vector $\mathbf{p}\in [\mathbf{0},\mathbf{l}]$ such that $\mathbf{p}=F(\mathbf{p})$, where

%\begin{equation*}
%F(\mathbf{x})_i= \left\{
%\begin{array}{ll}
%L_i, & \textrm{if } L_i\leq e_i+\sum_{j=1}^n x_j \pi_{ji}\\
%\alpha e_i+\beta \sum_{j=1}^n x_j\pi_{ji}, & \textrm{ otherwise}.\\
%\end{array} \right.
%\end{equation*}
%where $\pi_{ji}$ is the relative liability of $v_j$ towards $v_i$, i.e.,  $\pi_{ji}=\frac{l_{ji}}{L_j}$ if $L_j>0$, and $\pi_{ji}=0$ if $L_j=0$.

%We diverge from the common assumption of proportionality in payments, which dictates that a defaulting firm pays its creditors proportionally to their respective liabilities, following the work of Bertschinger et al. \cite{Hoefer19}.

\emph{Proportional payments} have been frequently studied in the financial literature (e.g., in \cite{Dem18,Eisenberg01,rogers2013failure}). Given clearing payments $\mathbf{P}$ \mk{(see Equation (\ref{eq:fixed-point}))}, each $p_{ij}$ must also satisfy $p_{ij} = \min\{l_{ij}, (e_i+\sum_{k\in[n]}{p_{ki}})\frac{l_{ij}}{L_i}\}. $
%\begin{equation*}
%p_{ij} = \min\{l_{ij}, (e_i+\sum_{k\in[n]}{p_{ki}})\frac{l_{ij}}{L_i}\}. 
%\end{equation*}
Note that, when constrained to use proportional payments, there is no strategic decision making involved.

We focus on \emph{priority-proportional payments}, where a firm's strategy is independent of its total assets and consists of a complete ordering of its creditors allowing for ties. Creditors of higher priority must be fully repayed before any payments are made towards creditors of lower priority, while creditors of equal priority are treated as in proportional payments. For example, a firm $v_i$ having firms $a,b,c,$ and $d$ as creditors may select strategy $s_i = (a,b|c|d)$ that has firms $a,b$ in the top priority class, followed by $c$ and, finally, with $d$ in the lowest priority class. 	

Let $L_i^{(m)}$ denote the total liability of firm $v_i$ to firms in its $m$-th priority class. We use parameters $k_{ij}$ to imply that firm $v_j$ is in the $k_{ij}$-th priority class of firm $v_i$ and denote by $\pi_{ij}'$ the relative liability of $v_i$ towards $v_j$ in the corresponding priority class, i.e.,  $\pi_{ij}'=\frac{l_{ij}}{L_{i}^{(k_{ij})}}$ if $L_{i}^{(k_{ij})}>0$, and $\pi_{ij}'=0$ if $L_{i}^{(k_{ij})}=0$. \mk{For given priority-proportional strategies for all firms, the clearing payments $\mathbf{P}$ (see Equation (\ref{eq:fixed-point}))} must also satisfy 
\begin{equation}\label{eq:pp-payments}
p_{ij} =\min \left\{ \max \left\{ \left(p_i-\sum_{m=1}^{k_{ij}-1} {L_i^{(m)}} \right)\cdot \pi_{ij}', 0 \right\} , l_{ij} \right\}. 
\end{equation}
That is, the payment of $v_i$ to a creditor in priority class $L_i^{(k_{ij})}$ occurs only after all payments to creditors of higher priority have been guaranteed. Then, payments to creditors in $L_i^{(k_{ij})}$ are made proportionally to their claims in that priority class. Finally, we also have $0\leq p_{ij}\leq l_{ij}$.

\mk{We will now define the notion of \emph{Nash equilibrium} in a financial network game. First, let us stress that a strategy profile has consistent clearing payments which are not necessarily unique. It is standard practice (see, e.g. \cite{Eisenberg01,rogers2013failure}) to focus attention to maximal clearing payments (such payments point-wise maximize all corresponding payments) to avoid this ambiguity. So, we say that a strategy profile $\mathbf{s}$ is a \emph{Nash equilibrium} if no firm can increase her utility by deviating to another payment strategy. We only consider pure Nash equilibria and clarify that the utility of a firm for a given strategy profile is computed based on the assumption that the maximal clearing payments will be realized every time. In Section \ref{sec:non-strategic}, we show how we can compute such payments efficiently.}
%Given a strategy profile $\mathbf{s}$ and a clearing vector $\mathbf{p}$ consistent with the strategy profile, if no firm can increase her utility by deviating to another payment strategy, then $\mathbf{p}$ paired with $\mathbf{s}$ form a \emph{pure Nash equilibrium}
 Extending strategy deviations to coalitions and joint deviations, we are interested in \emph{strong equilibria} where no coalition can cooperatively deviate so that all coalition members obtain strictly greater utility. %To simplify notation, in the following we only refer to the clearing vector $\mathbf{p}$ when arguing about equilibria. 

\paragraph{Social welfare:}
Given clearing payments $\mathbf{P}$, the social welfare $\sw(\mathbf{P})$ is the sum of the firm utilities; the particular utility notion (total assets or equities) will be clear from the context.
%For the utility notion of total assets, let $\sw_A(\p)=\sum_{i\in [n]}a_i(\p)$, while for the notion of equities, let $\sw_E(\p)=\sum_{i\in [n]}E_i(\p)$.  
The optimal social welfare is denoted by $\opt$.
%$\opt_A$, for total assets, and $\opt_E$, for equities, respectively. 
Let $\mathbf{P_{eq}}$ be the set of  \tc{clearing payments consistent with (pure) Nash equilibrium strategy profiles}. The price of anarchy (PoA) of a particular instance  is defined as the worst-case ratio of the optimal social welfare over the social welfare achieved at any equilibrium at the instance, $\poa=\max_{\pp \in \mathbf{P_{eq}}}\frac{\opt}{\sw(\mathbf{P})}.$
% \begin{displaymath}
%\poa=\max_{\p \in \mathbf{\p_{eq}}}\frac{\opt}{\sw(\p)}.
%\end{displaymath}
 In contrast, the price of stability (PoS) of a given instance of a game measures how far the highest social welfare that can be achieved at equilibrium is from the optimal social welfare, i.e., $\pos=\min_{\pp \in \mathbf{\pp_{eq}}}\frac{\opt}{\sw(\mathbf{P})}.$
% \begin{displaymath}
%\pos=\min_{\p \in \mathbf{\p_{eq}}}\frac{\opt}{\sw(\p)}.
%\end{displaymath}
The \emph{Price of Anarchy (Price of Stability, respectively)} of a game is the maximum PoA (PoS, respectively) of any instance of the given game.

\begin{example}\label{example}
We represent a financial network by a graph as follows. Nodes correspond to firms and black edges correspond to debt-liabilities; a directed edge from node $v_i$ to node $v_j$ with label $l_{ij}^0$ implies that firm $v_i$ owes firm $v_j$ an amount of money equal to $l_{ij}^0$. Nodes are also labeled, their label appears in a rectangle and denotes their external assets; we omit these labels for firms with external assets equal to $0$. A pair of red edges (one solid and one dotted) represents a CDS contract: a solid directed edge from node $v_i$ to node $v_j$ with label $l_{ij}^k$ and a dotted undirected edge connecting this edge with a third node $v_k$ implies that firm $v_i$ owes $v_j$ an amount of money equal to $(1-r_k)l_{ij}^k$. 

Figure \ref{fig:example} depicts a financial network with five firms having external assets $e_1=e_4=1$ and $e_2=e_3=e_5=0$. There exist four debt contracts, i.e., firm $v_1$ owes $v_2 $ and $v_3$ two coins and one coin, respectively; $v_2$ owes $v_1$ one coin, and $v_3$ owes $v_5$ one coin. There is also a CDS contract between $v_4$, $v_5$, and $v_3$, with nominal liability $l_{45}^3=1$, with $v_4$ being the debtor, $v_5$ the creditor, and $v_3$ the reference entity. $v_4$ will only need to pay $v_5$ if $v_3$ is in default; the amount owed would be equal to $1(1-r_3)$. 

\begin{figure}[h]
	\centering
	\includegraphics[width=0.38\textwidth, height=2.5cm]{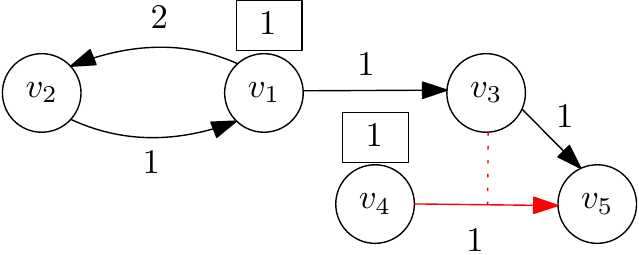}
	\caption{An example of a financial network.}
	\label{fig:example}
\end{figure}

In this network, only $v_1$ can strategize about its payments, hence we focus just on $v_1$'s strategy:
\begin{itemize}
\item Let $v_1$ select the priority-proportional strategy $s_1 = (v_2|v_3)$. Then, the payment vector would be $\mathbf{p}=(2, 1, 0, 1, 0)$ with $\p_1 = (0, 2, 0, 0, 0)$. Note that this is valid since $r_3=\frac{p_3}{L_3}=\frac{p_{35}}{l_{35}}=\frac{0}{1}=0$ and $p_{45}=(1-r_3)l^{3}_{45}=1$.  The total assets of the firms are $a_1(\mathbf{P})=2$, $a_2(\mathbf{P})=2$, $a_3(\mathbf{P})=0$, $a_4(\mathbf{P})=1$ and $a_5(\mathbf{P})=1$, and the social welfare is  $SW(\mathbf{P})=6$.
\item If $v_1$'s strategy is $s'_1 = (v_3|v_2)$, the only consistent payment vector would be $\mathbf{p'}=(1, 0, 1, 0, 0)$ with $\p'_1 = (0, 0, 1, 0, 0)$. Then, $a_1(\mathbf{P}')=1$, $a_2(\mathbf{P}')=0$, $a_3(\mathbf{P}')=1$, $a_4(\mathbf{P}')=1$ and $a_5(\mathbf{P}')=1$, resulting in $SW(\mathbf{p}')=4$.
\item If $v_1$ decides to pay proportionally, that is $s''_1=(v_2, v_3)$, the payment vector would be $\mathbf{p''}=(2, 1, 2/3, 1/3, 0)$ with $\p''_1 = (0, 4/3, 2/3, 0, 0)$; note that $p''_{45}=l^3_{45}(1-r_3)=1(1-\frac{p_{35}}{l_{35}})=1/3$. Then, $a_1(\mathbf{P}'')=2$, $a_2(\mathbf{P}'')=4/3$, $a_3(\mathbf{P}'')=2/3$, $a_4(\mathbf{P}'')=1$ and $a_5(\mathbf{P}'')=1$, resulting in $SW(\mathbf{P}'')=6$. 
\end{itemize}
Comparing the total assets under the different strategies discussed above, $v_1$ would select either strategy $(v_2|v_3)$ or   $(v_2, v_3)$, returning the maximum possible total assets (i.e., $2$). Therefore, any strategy profile $\s$ where either $s_1 = (v_2|v_3)$ or $s_1=(v_2, v_3)$ is a Nash equilibrium. 
\end{example}

\section{Existence and computation of clearing payments}\label{sec:non-strategic}

This section contains our results relating to the existence and the properties of \mk{(proper) clearing payments} in priority-proportional games.
We begin by arguing that, given a strategy profile, \mk{clearing payments} always exist even in presence of default costs. Furthermore, in case there are multiple clearing \mk{payments}, there exist maximal payments, i.e., ones that point-wise maximize all corresponding payments, and we provide a polynomial-time algorithm that computes them.  Note that this result is in a non-strategic context, but it is necessary in order to perform our game-theoretic analysis as it allows us to argue about well-defined deviations by considering clearing  \mk{payments} consistently among different strategy profiles.

\begin{lemma}\label{lem:max_cv_existence}
In priority-proportional games with default costs, there always exist \mk{maximal clearing payments} under a given strategy profile.
\end{lemma}

\begin{proof}
The proof follows by Tarski's fixed-point theorem, along similar lines to \cite{Eisenberg01,rogers2013failure}. We  first focus on the case of not necessarily proper payments. Indeed, the set of payments form a complete lattice. Any payment $p_{ij}$ is lower-bounded by $0$ and upper-bounded by $l_{ij}$ and for any two clearing payments $\pp$ and $\pp'$ such that $\pp \geq \pp'$ ($\pp$ is pointwise at least as big as $\pp'$) it holds that $\Phi(\pp)\geq \Phi(\pp')$, where $\Phi$ is defined in (\ref{eq:fixed-point}). Therefore, $\Phi(\cdot)$ has a greatest fixed-point and a least fixed-point.

\mk{We now claim that the existence of maximal clearing payments implies the existence of maximal proper clearing ones. Indeed, consider some maximal clearing payments $\pp$ and the resulting proper  payments $\pp'$ obtained by $\pp$ when ignoring all payments that do not originate from some firm with positive external assets. \tc{We note that $\pp'$ are clearing payments since the payments that  are deleted, in the first place only reached firms whose outgoing payments are decreased to zero as well.} Clearly, if $\pp'$ are not maximal proper clearing payments, then there exist proper clearing payments $\tilde{\pp}$ with $p'_{ij}<\tilde{p}_{ij}$ for some firms $v_i, v_j$. If $p'_{ij} = p_{ij}$ we obtain a contradiction to the maximality of $\pp$, otherwise if $p'_{ij}<p_{ij}$, then $\tilde{\pp}$ cannot be proper. }
\end{proof}

%\begin{corollary}\label{cor:existence}
%In priority-proportional games with default costs, there always exist maximal clearing payments under a given strategy profile.
%\end{corollary}

We now show how such maximal \mk{clearing payments} can be computed. Given a strategy profile in a priority-proportional game, Algorithm \ref{alg:MCV} below, that extends related algorithms in \cite{Eisenberg01,rogers2013failure}, computes the maximal (proper) clearing \mk{payments} in polynomial time. \mk{In particular, and since the strategy profile is fixed, we will argue about the payment vector consisting of the total outgoing payments for each firm; the detailed payments then follow by the strategy profile.}
%Recall that $L_i^{(m)}$ denotes the total liability that firm $v_i$ has towards the firms in its $m$-th priority class,
%and that %$TL_i^{(\underline{k})}$ denotes the total liability that firm $v_i$ has towards the firms in its priority classes $1$ up to $k-1$, i.e. $TL_i^{(\underline{k})}=\sum_{m=1}^{k-1} {L_i^{(m)}}$. Also,
%$\pi_{ij}'=\frac{l_{ij}}{L_i^{(k_{ij})}}$ is the relative liability of $v_i$ towards $v_j$ in the corresponding priority class ($v_j$ is in the $k_{ij}$-th priority class of firm $v_i$). For given strategies (priority classes), recall that the clearing vector $\mathbf{p}$ has to satisfy, for all pairs of firms $v_i$, $v_j$, that the payment from $v_i$ to $v_j$ is $p_{ij} =\min \left\{ \max \left\{ \left(p_i-\sum_{m=1}^{k_{ij}-1} {L_i^{(m)}} \right)\cdot \pi_{ij}', 0 \right\} , l_{ij} \right\}.$

%\vspace{0.5cm}

\begin{algorithm}[h]
		\caption{\textsc{MCP}}\label{alg:MCV}
		\tcc{The algorithm assumes given strategies (priority classes). By abusing notation and for ease of exposition, we denote by $p_i^{(\kappa)}$ the \mk{total} outgoing payments of firm $v_i$ at round $\kappa$ and by $\p^{(\kappa)}$ the vector of \mk{total} outgoing payments at round $\kappa$. }
		%The payment vector computed at each round satisfies Equation (\ref{eq:pp-payments}). }
	Set $\mu=0, \mathbf{p}^{(-1)}=\mathbf{l}$ and $ \ds_{-1}=\emptyset$\;
	Compute $ E_i^{(\mu)}:=e_i+\displaystyle\sum_{j\in [n]} p^{(\mu-1)}_{ji}-L_i $,  for $i=1,\ldots, n$\; \label{alg:line2}
	$\ds_{\mu} =\left\{  v_i: E_i^{(\mu)} < 0 \right\} $\;

	\If{$\ds_{\mu} \neq \ds_{\mu-1}$}{
	Compute $\p^{(\mu)}$ that is consistent with Equation (\ref{eq:pp-payments}) and satisfies $p_{i}^{(\mu)}=\left\{
\begin{array}{ll}
\alpha e_i + \beta \left(\displaystyle\sum_{j \in \ds_{\mu}  } p^{(\mu)}_{ji}+ \displaystyle\sum_{ j \in \mathcal{N}\setminus \ds_{\mu} } l_{ji}  \right),& \forall v_i \in \ds_{\mu}\\
L_{i} & \forall v_i \in \mathcal{N}\setminus\ds_{\mu}.\\
\end{array} \right.$\;
	Set $\mu = \mu+1$\;
	\textbf{go to} Line \ref{alg:line2}\;}
\Else{Run \textsc{Proper}($\mathbf{p}^{(\mu-1)}$)}
\end{algorithm}

\begin{algorithm}[H]
\caption{\textsc{Proper($\mathbf{x}$)}}\label{alg:proper}
\tcc{The algorithm takes in input payments $\mathbf{x}$ and returns proper payments.}
Set \textsc{Marked} $=\{v_i: e_i>0\}$ and \textsc{Checked} $= \emptyset$\;
\While{\textsc{Marked} $\not=\emptyset$}{
Pick $v_i\in \textsc{Marked}$\; 
\For{$v_j \not\in$ \textsc{Marked} $\cup$ \textsc{Checked} with $x_{ij}>0$}
{\textsc{Marked} $=$ \textsc{Marked} $\cup \{v_j\}$\;}
\textsc{Marked} $=$ \textsc{Marked} $\setminus \{v_i\}$ and \textsc{Checked} $=$ \textsc{Checked} $\cup \{v_i\}$\;
}
\For{$v_i\notin$ \textsc{Checked}}{Set all outgoing payments from $v_i$ in $\textbf{x}$ to $0$\;}
Return $\mathbf{x}$
\end{algorithm}

\begin{lemma}\label{lem:induction}
The payment vectors computed in each round of Algorithm \ref{alg:MCV} are pointwise non-increasing, i.e., $\p^{(\mu)}\leq \p^{(\mu-1)}$ for any round $\mu\geq 0$.
\end{lemma}
\begin{proof}
We prove the lemma by induction. The base of our induction is $\p^{(0)}\leq \p^{(-1)}=\mathbf{l}$. It holds that $p_i^{(0)}=L_i$ if $v_i\in \mathcal{N}\setminus \ds_{0}$, so it suffices to compute $p_i^{(0)}$ and show that $p_i^{(0)}\leq L_i$  for $v_i\in \ds_0$.

We wish to find the solution $\mathbf{x}$ to the following system of equations
\begin{eqnarray}\nonumber
&x_i=\alpha e_i + \beta \left( \sum_{j \in \ds_{0}} x_{ji}+ \sum_{ j \in \mathcal{N}\setminus \ds_{0} } l_{ji}   \right), &\qquad\forall  v_i \in \ds_{0}, \\\label{con:payment-in-default}
&x_i=L_i, &\qquad\forall  v_i \in \mathcal{N}\setminus \ds_{0}.
\end{eqnarray}
We compute $\mathbf{x}$ using a recursive method and starting from $\mathbf{x}^{(0)}=\mathbf{p}^{(-1)}=\mathbf{l}$. We define $x^{(k)}$, $k\geq 1$, recursively by
\begin{equation*}
x_i^{(k+1)}=\alpha e_i + \beta \left( \sum_{j \in \ds_{0}} x^{(k)}_{ji}+ \sum_{ j \in \mathcal{N}\setminus\ds_{0} } l_{ji}   \right).
\end{equation*}
Now for $v_i \in \ds_{0}$, we have

$$ x_i^{(1)} = \alpha e_i + \beta \left( \sum_{j \in \ds_{0}} x^{(0)}_{ji}+ \sum_{ j \in \mathcal{N}\setminus\ds_{0} } l_{ji}   \right)
 \leq e_i+ \sum_{j \in \ds_{0}} x^{(0)}_{ji}+ \sum_{ j \in \mathcal{N}\setminus\ds_{0} } l_{ji}
< L_i=x_i^{(0)},$$
%\begin{align*}
 %x_i^{(1)} &= \alpha e_i + \beta \left\{ \sum_{j \in \ds_{0}} P_{ji}\left(\mathbf{x}^{(0)}\right)+ \sum_{ j \in \mathcal{N}\setminus\ds_{0} } l_{ji}   \right\} \\
 %&\leq e_i+ \left\{ \sum_{j \in \ds_{0}} P_{ji}\left(\mathbf{x}^{(0)}\right)+ \sum_{ j \in \mathcal{N}\setminus\ds_{0} } l_{ji}   \right\} \\
 %& \leq e_i+\left\{ \sum_{j \in \ds_{0}} l_{ji}+ \sum_{ j \in \mathcal{N}\setminus\ds_{0} } l_{ji}   \right\}\\
 %&=e_i+\sum_{j}^n l_{ji} \\
 %&< L_i=x_i^{(0)},
%\end{align*}
where the first inequality holds because $\alpha, \beta \leq 1$, and the second inequality holds by our assumption that $v_i \in \ds_{0}$.
Hence, sequence $ \mathbf{x}^{(k)}$ is decreasing. Since the solution to Equation (\ref{con:payment-in-default}) is non-negative, $\mathbf{x}$ can be computed as $\mathbf{x}=\lim_{k \rightarrow \infty }\mathbf{x}^{(k)}$, which completes the base of our induction.

Now assume that $ \mathbf{p} ^{(\mu)} \leq \mathbf{p} ^{(\mu-1)}$ for some $\mu\geq 0$. We will prove that $\mathbf{p}^{(\mu+1)} \leq \mathbf{p}^{(\mu)}$.
Similarly to before, $p_i^{(\mu+1)}=p_i^{(\mu)}=L_i$ if $v_i\in \mathcal{N}\setminus \ds_{\mu+1}$, so it suffices to compute $p_i^{(\mu+1)}$ and show that $p_i^{(\mu+1)}\leq p_i^{(\mu)}$  for $v_i\in \ds_{\mu+1}$.

The desired $p_i^{(\mu+1)}$ is the solution $\mathbf{x}$ to the following system of equations

\begin{align}\label{mu-solution}
x_i=\left\{
\begin{array}{ll}
\alpha e_i + \beta \left( \displaystyle\sum_{j \in \ds_{\mu+1}} x_{ji}+ \displaystyle\sum_{ j \in \mathcal{N}\setminus\ds_{\mu+1} } l_{ji}   \right),& \qquad\forall  i \in \ds_{\mu+1},\\
L_i, &\qquad\forall  v_i \in \mathcal{N}\setminus \ds_{\mu+1}.
\end{array} \right.
\end{align}

We compute $\mathbf{x}$ recursively starting with $ \mathbf{x}^{(0)}=\mathbf{p}^{(\mu)}$. We define $x^{(k)}$, $k\geq 1$, recursively by
\begin{equation}
x_i^{(k+1)}=\alpha e_i + \beta \left( \sum_{j \in \ds_{\mu+1}} x^{(k)}_{ji}+ \sum_{ j \in \mathcal{N}\setminus\ds_{\mu+1} } l_{ji}   \right).
\label{mu-recrusive}
\end{equation}

For $v_i \in \ds_{\mu+1}$, we have
\begin{align*}
x_i^{(1)}&= \alpha e_i + \beta \left( \sum_{j \in \ds_{\mu+1}} x^{(0)}_{ji}+ \sum_{ j \in \mathcal{N}\setminus\ds_{\mu+1} } l_{ji}   \right)\\
& =\alpha e_i + \beta \left( \sum_{j \in \ds_{\mu+1}} p^{(\mu)}_{ji}+ \sum_{ j \in \mathcal{N}\setminus\ds_{\mu+1} } l_{ji}   \right)\\
&=\alpha e_i + \beta \left( \left( \sum_{j \in \ds_{\mu}} p^{(\mu)}_{ji}+ \sum_{j \in \ds_{\mu+1}\setminus\ds_{\mu}} p^{(\mu)}_{ji}\right)+\sum_{ j \in \mathcal{N}\setminus\ds_{\mu+1} } l_{ji}   \right)\\
&=\alpha e_i + \beta \left( \left( \sum_{j \in \ds_{\mu}} p^{(\mu)}_{ji}+ \sum_{j \in \ds_{\mu+1}\setminus\ds_{\mu}} l_{ji} \right)+\sum_{ j \in \mathcal{N}\setminus\ds_{\mu+1} } l_{ji}   \right)\\
&=\alpha e_i + \beta \left(  \sum_{j \in \ds_{\mu}} p^{(\mu)}_{ji}+ \sum_{ j \in \mathcal{N}\setminus\ds_{\mu} } l_{ji}   \right),
%\label{con:payment-in-default}
\end{align*}
where we note that our assumption $ \mathbf{p} ^{(\mu)} \leq \mathbf{p} ^{(\mu-1)}$ implies that $\ds_{\mu+1} \supseteq  \ds_{\mu}$.
Now we can split the set $\ds_{\mu+1}$ into  $\ds_{\mu}$ and $\ds_{\mu+1} \setminus \ds_{\mu}$. For $v_i \in  \ds_{\mu}$ we have
$x_i^{(1)}=\mathbf{p}^{(\mu)}=x_i^{(0)}$. For $v_i\in \ds_{\mu+1} \setminus \ds_{\mu}$ we have

\begin{align*}
x_i^{(1)}&=\alpha e_i + \beta \left(  \sum_{j \in \ds_{\mu}} p^{(\mu)}_{ji}+ \sum_{ j \in \mathcal{N}\setminus\ds_{\mu} } l_{ji}   \right)\\
&\leq e_i +  \sum_{j \in \ds_{\mu}} p^{(\mu)}_{ji}+ \sum_{ j \in \mathcal{N}\setminus\ds_{\mu} } l_{ji}  \\
&=e_i +  \sum_{j\in [n]}{p^{(\mu)}_{ji}} \\
&< L_i = {p}^{(\mu)}_i=x_i^{(0)},
\end{align*}
which implies that the sequence $x^{(k)}$ is decreasing. Since the solution to Equation (\ref{mu-solution}) is non-negative, $\mathbf{x}=\mathbf{p}^{\mu+1}$ can be computed as $\mathbf{x}=\lim_{k \rightarrow \infty }\mathbf{x}^{(k)}$, which completes our claim that $\p^{(\mu)}\leq \p^{(\mu-1)}$ for any round $\mu\geq 0$.
\end{proof}

We now present the main result of this section.
\begin{theorem}\label{thm:maxcv-algorithm}
Algorithm \ref{alg:MCV} computes the maximal clearing payments under priority-proportional strategies in polynomial time.
\end{theorem}

\begin{proof}
The algorithm proceeds in rounds. In each round $\mu$, tentative vectors of payments, $\p^{(\mu)}=(p_1^{(\mu)},\ldots,p_n^{(\mu)})$, and effective equities, $E^{(\mu)}$,  are computed. At the beginning of round $0$, all firms are marked as tentatively solvent, which we denote by $\ds_{-1}=\emptyset$; $\ds_{\mu}$  is used to denote the firms in default after the $\mu$-th round of the algorithm. The algorithm works so that once a firm is in default in some round, then it remains in default until the termination of the algorithm. Indeed, by Lemma \ref{lem:induction} the vectors of payments are non-increasing between rounds and the strategies are fixed. \mk{Algorithm \textsc{Proper} is called when $\ds_{\mu}=\ds_{\mu-1}$, which requires at most $n$ rounds; clearly, each round requires polynomial time. The running time of \textsc{Proper} is also polynomial. Indeed, note that each firm can enter set \textsc{Marked} at most once and will leave \textsc{Marked} to join set \textsc{Checked}} \mk{after each other firm is examined at most once.} 

Regarding the correctness of the Algorithm, we start by proving by induction that the \mk{payment vector} provided as input to Algorithm \ref{alg:proper} is at least equal (pointwise) to the maximal clearing vector $\mathbf{p^{*}}$. As a  base of our induction, it is easy to see that $\mathbf{p}^{(-1)}=\mathbf{l} \geq \mathbf{p}^*$. Now assume that $\mathbf{p}^{(\mu-1)} \geq \mathbf{p}^*$ for some $\mu\geq 0$; we will prove that $\mathbf{p}^{(\mu)} \geq \mathbf{p}^*$. We denote by $\ds_{*}$ the firms in default under the maximal clearing vector $\mathbf{p}^{*}$, i.e., $
\ds_{*}=\left \{v_i:e_i+\sum_{j\in[n]}{ p^{*}_{ji}} <L_{i} \right\}$. Our inductive hypothesis $\mathbf{p}^{(\mu-1)} \geq \mathbf{p}^*$ implies $\ds_{\mu} \subseteq \ds_{*}$. Hence, for firms $v_i \in \mathcal{N}\setminus\ds_{\mu}$, we have $\mathbf{p}_i^{(\mu)}=L_i \geq p_i^*$. For  $v_i \in \ds_{\mu}$ we refer to the proof of Lemma \ref{lem:induction} above and consider Equation (\ref{mu-recrusive}) again while starting the recursive solution with $x_i^{(0)}={p}_i^{(\mu)}$. For $v_i\in \ds_{\mu}$, we observe
\begin{align*}
x_i^{(1)}&=\alpha e_i + \beta \left( \sum_{j \in \ds_{\mu}} p^{(\mu)}_{ji}+ \sum_{ j \in \mathcal{N}\setminus\ds_{\mu} } l_{ji}   \right)\\
&\geq \alpha \cdot e_i + \beta \cdot \left(  \sum_{j \in \ds_{\mu}}  p^*_{ji}+ \sum_{ j \in \mathcal{N}\setminus\ds_{\mu} } l_{ji}   \right)\\
&\geq \alpha e_i+ \beta \left( \sum_{j\in[n]}{p^*_{ji}}\right)\\
&=L_i^{(0)}.
\end{align*}
Recursion (\ref{mu-recrusive}) then implies that $x^{(k)} \geq \p^{*}$ for all $k$ and hence $\mathbf{p}^{(\mu+1)}=\mathbf{x}=\lim_{k \rightarrow \infty}\mathbf{x}^{k} \geq \mathbf{p}^{*}$.

We have now proved that the input to Algorithm \ref{alg:proper} is at least equal (pointwise) to the desired maximal clearing vector $\mathbf{p^{*}}$. However, by Lemma \ref{lem:induction} we also know that $\mathbf{p}^{(\mu)} \leq \mathbf{p}^{(\mu-1)}$ for all $\mu\geq 0$. It holds by design that the input of Algorithm \ref{alg:proper} is a clearing vector, so $\mathbf{p}^{*}$ is the only possible such input. By the arguments in the proof of Lemma \ref{lem:max_cv_existence}, Algorithm \ref{alg:proper} with input the maximal clearing payments computes the maximal proper clearing payments, and the claim follows.
\end{proof}

%\cleardoublepage
\section{Maximizing the total assets}\label{sec:assets} 
We now turn our attention to financial network games under priority-proportional strategies when the utility is defined as the total assets. We note that in this case, the maximal \mk{clearing payments}, computed in Section \ref{sec:non-strategic} are weakly preferred by all firms among all clearing payments of the given strategy profile; indeed, the utility is computed as the sum between a fixed term (external assets) and the incoming payments which are by definition maximized. So, in case of \mk{various} clearing payments, it is reasonable to limit our attention to the (unique) maximal clearing payments computed in Section~\ref{sec:non-strategic}.

%\subsection{Existence and properties of equilibria}
We begin with a negative result regarding the existence of Nash equilibria. 

\begin{theorem}\label{thm:pp-no-nash-equilibria}
Nash equilibria are not guaranteed to exist when firms aim to maximize their total assets.
\end{theorem}

\begin{proof-sketch}
\begin{figure}[h]
\centering
\includegraphics[height=2.2cm]{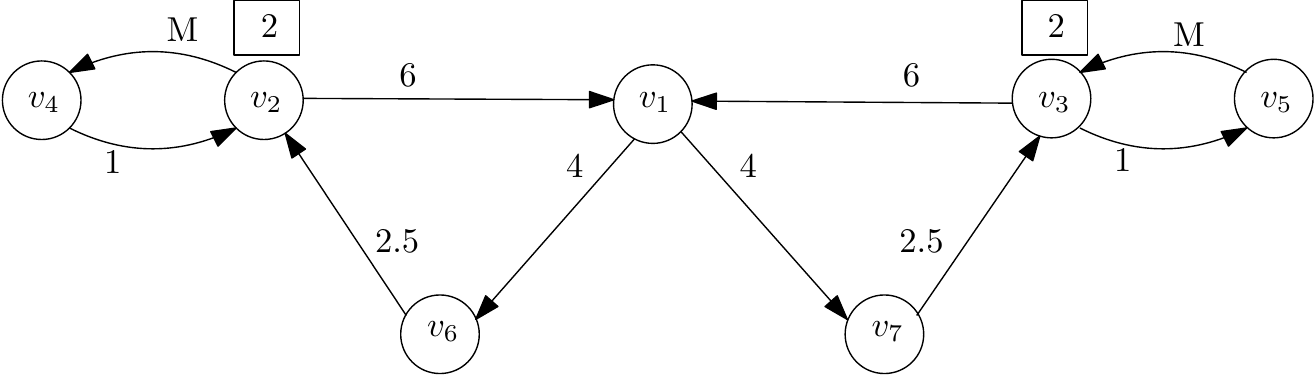}
\caption{A game that does not admit Nash equilibria.}\label{fig:no_NE}
\end{figure}
Consider the financial network depicted in Figure \ref{fig:no_NE} where $M$ is an arbitrarily large integer. Only firms $v_1$, $v_2$, $v_3$ have more than one available strategies and, hence, it suffices to argue about them. \tc{The instance is inspired by an equivalent result in \cite{Hoefer19} regarding edge-ranking strategies, however the instance used in the proof of that result admits an equilibrium under priority-proportional strategies.}

Observe that since $M$ is large, whenever $v_2$ (respectively, $v_3$) has $v_4$ (respectively, $v_5$) in its top priority class, either alone or together with $v_1$, then the payment towards $v_1$ is at most $5.5\varepsilon$, where $\varepsilon = \frac{6}{M+6}$, i.e., a very small payment. Furthermore, $v_2$ and $v_3$ can never fully repay their liabilities to any of their creditors; this implies that any (non-proportional) strategy that has a single creditor in the topmost priority class will never allow for payments to the second creditor.

If none of $v_2$ and $v_3$ has $v_1$ as their single topmost priority creditor, then, by the remark above about payments towards $v_1$, at least one of them has an incentive to deviate and set $v_1$ as their top priority creditor. Indeed, each of $v_2$, $v_3$ has utility at most $\frac{3+3\varepsilon}{1-\varepsilon}$, and at least one of them would receive utility at least $4$ by deviating.
Furthermore, it cannot be that both $v_2$ and $v_3$ have $v_1$ as their single top priority creditor, as at least one of them is in the top priority class of $v_1$. This firm, then, wishes to deviate and follow a proportional strategy so as to receive also the payment from its other debtor.

It remains to consider the case where one of $v_2$, $v_3$, let it be $v_2$, has $v_1$ as the single top priority creditor and the remaining firm, let it be $v_3$, either has a proportional strategy or has $v_5$ as its single top priority creditor. In this setting, if $v_1$ follows a proportional strategy, $v_1$ wishes to deviate and select $v_2$ as its single top priority creditor. Otherwise, if $v_1$ has $v_2$ as its top priority creditor, then $v_3$ deviates and sets $v_1$ as its top creditor. Finally, if $v_1$ has $v_3$ as its single top priority creditor, then $v_3$ deviates to a proportional strategy.

The full table of utilities appears in Appendix \ref{app:no-nash-equilibria} (Table \ref{tab:inexistence}).
\end{proof-sketch}

%\subsection{(In)efficiency of equilibria} 
We now aim to quantify the social welfare loss in Nash equilibria when each firm aims to maximize its total assets. While the focus is on financial network games under priority-proportional payments, we warm-up by considering the well-studied case of proportional payments, where we show that these may lead to outcomes where the social welfare can be far from optimal. 

\begin{theorem}\label{thm:proportional-inefficiency}
Proportional payments can lead to arbitrarily bad social welfare loss with respect to total assets. In acyclic financial networks, the social welfare loss is at most a factor of $n/2$ and this is almost tight.
\end{theorem}
\begin{proof}
Consider the financial network between firms  $v_1$, $v_2$, and $v_3$ that is shown in Figure \ref{fig:prop-asset-unbounded}, where $M$ is arbitrarily large.
\begin{figure}[h]
\centering
\includegraphics[height=1.5cm]{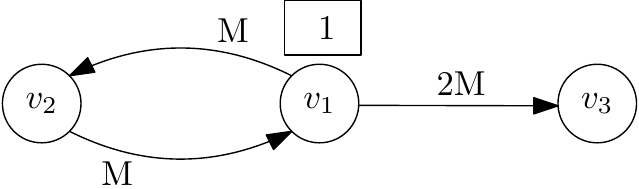}
\caption{A financial network where proportional payments lead to low social welfare with respect to total assets. $M$ is arbitrarily large.}
\label{fig:prop-asset-unbounded}
\end{figure}
Observe that paying proportionally leads to clearing payments $\mathbf{p}_1=(0, 1/2, 1),  \mathbf{p}_2=(1/2, 0, 0), \mathbf{p}_3=(0, 0, 0)$. Hence,  the  total assets are $a_1(\pp)=3/2$, $a_2(\pp)=1/2$, and $a_3(\pp)=1$, and, therefore, $\sw(\pp)=3$. However, if firm $v_1$ chooses to pay firm $v_2$, the resulting clearing payments would be $\mathbf{p}'_1=(0, M, 1),  \mathbf{p}'_2=(M, 0, 0), \mathbf{p}'_3=(0, 0, 0)$ with total assets $a_1(\pp')=1+M$, $a_2(\pp')=M$, and $a_3(\mathbf{P}')=1$, that sum up to $\sw(\pp')=2M+2$. Since $\opt\geq \sw(\mathbf{P}')$ and $M$ can be very large, we conclude that the social welfare achieved by proportional payments can be arbitrarily smaller than the optimal social welfare.

For the case of acyclic financial networks, let $\xi_1$ and $\xi_2$ denote the total external assets of non-leaf and leaf nodes, respectively; note that a leaf node has no creditors in the network. Clearly, since there are no default costs, a non-leaf firm $v_i$ with external assets $e_i$ and total liabilities $L_i$ will generate additional revenue of at least $\min\{e_i, L_i\}$ through payments to its neighboring firms. Hence, for any clearing payments $\pp$  we have that $\sw(\mathbf{P}) \geq \xi_1+\sum_i{\min\{e_i, L_i\}} + \xi_2$. In the optimal clearing  payments, each non-leaf firm $v_i$ with external assets $e_i$ may generate additional revenue of at most $(n-1)\cdot\min\{e_i, L_i\}$, since the network is acyclic, and, therefore, we have that $\opt\leq \xi_1+(n-1)\sum_i{\min\{e_i, L_i\}}+\xi_2$, i.e., the social welfare loss is at most a factor of $n/2$ as the ratio is maximized when, for each $v_i$ we have $e_i=L_i$.

To see that this social welfare loss factor is almost tight, consider an acyclic financial network where firm $v_1$ with external asset $e_1=1$ has two creditors, $v_2$ and $v_3$, with $l_{12} = l^0_{12} =M$ and $l_{13} = l^0_{13} = 1$, where $M$ is an arbitrarily large integer. Firm $v_3$ is, then, the first firm along a path from $v_3$ to $v_n$, where for $i\in\{4,\dots, n\}$ $v_i$ is a creditor of $v_{i-1}$ and all liabilities equal $1$. Under proportional payments, we obtain that  $\mathbf{p}_1=(0, M/(M+1), 1/(M+1), \dots,  0), \mathbf{p}_2=\mathbf{p}_n=(0, \dots,  0) $ and for any $i \in \left\lbrace 3, \dots, n-1\right\rbrace $, $\mathbf{p}_i=(0, \dots, 1/(M+1), \dots, 0)$ where $1/(M+1)$ is the $(i+1)^{th}$ entry, thus  $\sw(\pp) = 2+(n-3)/(M+1)$. In the optimal clearing payments, firm $v_1$ fully repays its liability towards $v_3$ and this payment propagates along the path from $v_3$ to $v_n$ resulting to $\opt = n-1$, leading to a social welfare loss factor of $(n-1)/2-\varepsilon$, where $\varepsilon$ goes to zero as $M$ tends to infinity.
\end{proof}

We remark that, given clearing  payments with proportional payments, a firm may wish to deviate. 
%Indeed, consider the instance shown in Figure \ref{fig:prop-asset-unbounded} and the proof of Theorem \ref{thm:proportional-inefficiency}; firm $v_1$ has more total assets when playing strategy $(v_2|v_3)$ than when using the proportional strategy $(v_2,v_3)$.
\begin{remark}\label{remark:prop-not-ne}
Proportional payments may not form a Nash equilibrium when firms aim to maximize their total assets.
\end{remark}
\begin{proof}
Consider the instance shown in Figure \ref{fig:prop-asset-unbounded} and the proof of Theorem \ref{thm:proportional-inefficiency}; firm $v_1$ has more total assets when playing strategy $(v_2|v_3)$ than when using the proportional strategy $(v_2,v_3)$.
\end{proof}

We now turn our attention to priority-proportional strategies. To avoid text repetitions, we omit referring to priority-proportional games in our statements.
%When we consider financial networks with CDSs, clearly all negative results still hold. Our next theorem shows that even the best Nash equilibrium may be inefficient in this setting, contrary to the case of debt-only contracts studied in \cite{Hoefer19}.
We start with a positive result on the quality of equilibria when allowing for default costs in the (extreme) case $\alpha=\beta=0$.
\begin{theorem}\label{thm:pp-pos-default}
The price of stability is $1$ if default costs $\alpha=\beta=0$ apply and firms aim to maximize their total assets.
\end{theorem}
\begin{proof}
Consider the \mk{equilibrium resulting to the optimal social welfare}. Clearly, each solvent firm pays all its liabilities to its creditors and, hence, its strategy is irrelevant. Furthermore, each firm that is in default cannot make any payments towards its creditors, due to the default costs, therefore its strategy is irrelevant as well. Since no firm can increase its total assets, any strategy profile \mk{that leads to optimal social welfare} is a Nash equilibrium and the theorem follows.
\end{proof}

In the more general case, however, the price of stability may be unbounded, as the following result suggests. Recall that the setting without default costs corresponds to $\alpha=\beta=1$.
\begin{theorem}\label{thm:pp-pos-positivedefault}
The price of stability is unbounded if default costs $\alpha>0$ or $\beta >0$ apply when firms aim to maximize their total assets.
\end{theorem}
\begin{proof}
We begin with the case where $\beta>0$ and consider the financial network shown in Figure \ref{fig:alpha} where $M$ is arbitrarily large. Firm $v_1$ is the only firm that can strategize about its payments. Its strategy set comprises $(v_2|v_3)$,  $(v_3|v_2)$, and $(v_2,v_3)$, which result in utility $1+\beta$, $1$, and $1+\frac{2\beta^3}{M+2\beta-2\beta^3}$, respectively; note that, unless $v_1$ selects strategy $(v_2| v_3)$, $v_2$ is also in default as $M$ is arbitrarily large. For sufficiently large $M$, we observe that any Nash equilibrium must have $v_1$ choosing strategy $(v_2|v_3)$, leading to clearing payments $\mathbf{p}_1=(0, \beta^2+\beta, 0, 0, 0), \mathbf{p}_2=(\beta, 0, 0, 0, 0), \mathbf{p}_3=\mathbf{p}_4=(0, 0, 0, 0, 0)$, and $\mathbf{p}_5=(1, 0, 0, 0, 0)$ with $\sw(\pp)=2+\beta^2+2\beta$. Now, when $v_1$ chooses strategy $(v_3|v_2)$, we obtain clearing payments $\mathbf{p}'_1=(0, 0, \beta, 0, 0), \mathbf{p}'_2=(0, 0, 0, 0, 0), \mathbf{p}'_3=(0, 0, 0, M, 0), \mathbf{p}'_4=(0, 0, M, 0, 0)$, and $\mathbf{p}'_5=(1, 0, 0, 0, 0)$ with $\sw(\pp') =2M+2+\beta$. The claim follows since $\opt\geq \sw(\pp')$.

Now, let us assume that $\alpha>0$ and $\beta=0$ and consider the financial network shown in Figure~\ref{fig:beta}.
\begin{figure}[h]
\centering
\subfigure[The case where $\beta>0$]{
\centering
\includegraphics[width=5cm]{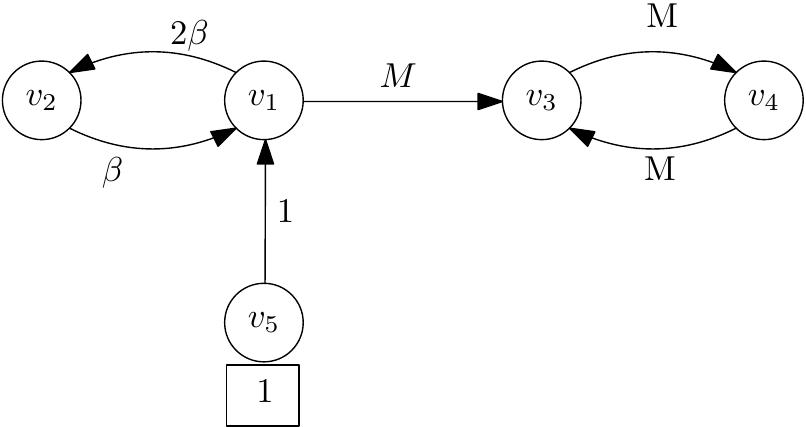}
\label{fig:alpha}}
\quad
\subfigure[The case where $\alpha>0$ and $\beta=0$]{
\centering
\includegraphics[width=5cm]{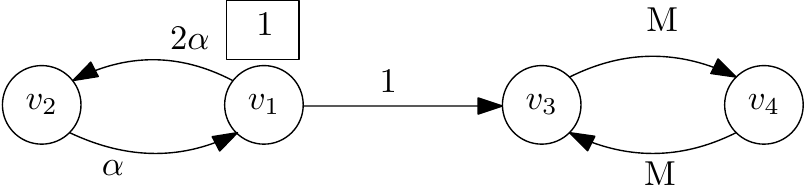}
\label{fig:beta}}
\caption{The instances used in the proof of the unbounded price of stability.}\label{fig:union-figure}
\end{figure}

Again, firm $v_1$ is
the only firm that can strategize about its payments. $v_1$'s total assets when choosing strategies $(v_2|v_3)$, $(v_3|v_2)$, and $(v_2, v_3)$ are $1+\alpha$, $1$, and $1$, respectively; note that when $v_1$ chooses strategy $(v_2,v_3)$, firm $v_2$ is also in default as it receives a payment of $\frac{2\alpha^2}{1+2\alpha}$ which is strictly less than $\alpha$ for any $\alpha>0$. Hence, in any Nash equilibrium $v_1$ chooses strategy $(v_2|v_3)$, resulting in clearing payments $\mathbf{p}_1=(0, \alpha, 0, 0), \mathbf{p}_2=(\alpha, 0, 0, 0), \mathbf{p}_3=\mathbf{p}_4=(0, 0, 0, 0)$ with social welfare $\sw(\pp)=1+2\alpha$. The optimal social welfare, however, is achieved when $v_1$ chooses strategy $(v_3|v_2)$, resulting in  clearing payments $\mathbf{p}'_1=(0, 0, \alpha, 0), \mathbf{p}'_2=(0, 0, 0, 0), \mathbf{p}'_3=(0, 0, 0, M)$, and $\mathbf{p}'_4=(0, 0, M, 0)$ with $\opt = 2M+1+\alpha$.
\end{proof}

We now show that the price of stability may also be unbounded in the absence of default costs, if negative external assets are allowed.
\begin{theorem}\label{thm:pp-pos-negative}
The price of stability is unbounded if negative external assets are allowed and firms aim to maximize their total assets.
\end{theorem}

\begin{proof}
Consider the financial network shown in Figure \ref{fig:negative-asset}.
\begin{figure}[h]
\centering
\includegraphics[height = 1.6cm]{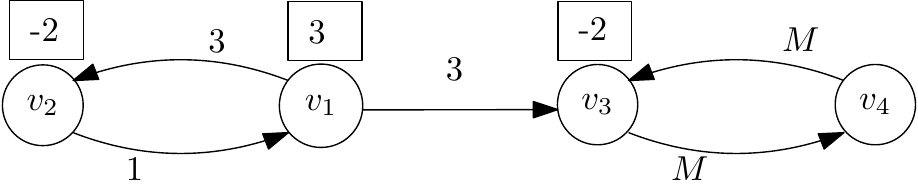}
\caption{A financial network with unbounded price of stability.}
\label{fig:negative-asset}
\end{figure}
Only firm $v_1$ can strategize about its payments. Clearly, $v_1$'s total assets equal $3$, unless $v_2$ pays its debt (even partially). This is only possible when $v_1$ chooses strategy $(v_2|v_3)$ and prioritizes  the payment of $v_2$ (note $v_2$'s negative external assets will ``absorb'' any payment that is at most $2$). Therefore, the resulting Nash equilibrium leads to the clearing payments $\mathbf{p}_1=(0, 3, 1, 0),  \mathbf{p}_2=(1, 0, 0, 0), \mathbf{p}_3=(0, 0, 0, 0)$, and $\mathbf{p}_4=(0, 0, 0, 0)$ with social welfare $\sw(\pp) = 4$. However, when $v_1$ chooses strategy $(v_3|v_2)$ we obtain the clearing payments $\mathbf{p}'_1=(0, 0, 3, 0),\mathbf{p}'_2=(0, 0, 0, 0), \mathbf{p}'_3=(0, 0, 0, M)$, and $\mathbf{p}'_4=(0, 0, M, 0)$ with social welfare $\sw(\pp') = 2M+2$. Hence, since $\opt\geq \sw(\pp')$ we obtain the theorem.
\end{proof}

The proof of Theorem \ref{thm:pp-pos-negative} in fact holds for any type of strategies as $v_1$ always prefers to pay in full its liability to $v_2$. This includes the case of a very general payment strategy scheme, namely coin-ranking strategies  \cite{Hoefer19} that are known to have price of stability $1$ with non-negative external assets.

Next, we show that the price of anarchy can be unbounded even in the absence of default costs, CDS contracts, and negative externals.
% even for financial networks with proper clearing vectors, where for any firm $v_i$ that makes positive payments towards its creditors there is a directed path originating from a node possessing external assets and reaching $v_i$. 
Bertschinger et al. \cite{Hoefer19} have shown a similar result for coin-ranking strategies, albeit for a network that has no external assets; our result extends to the case of coin-ranking strategies and strengthens the result of \cite{Hoefer19} to capture the case of proper \mk{clearing payments}. 

\begin{theorem}\label{thm:crazy}
The price of anarchy is unbounded when firms aim to maximize their total assets.
\end{theorem}
\begin{proof}
Consider the financial network between firms $v_i$, $i\in [4]$, that is shown in Figure \ref{fig:coin-regular-assets}, where $M$ is arbitrarily large.
\begin{figure}[h]
\centering
\includegraphics[height=2.5cm]{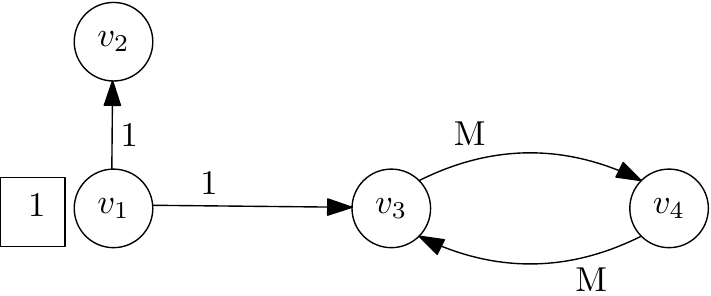}
\caption{A financial network that yields unbounded price of anarchy with respect to total assets.}
\label{fig:coin-regular-assets}
\end{figure}
Clearly, firm $v_1$ is the only firm that can strategize about its payments, and observe that $a_1 = 1$ regardless of $v_1$'s strategy. Hence, any strategy profile in this game is a Nash equilibrium. Consider the clearing payments $\mathbf{p}_1=(0, 1, 0, 0), \mathbf{p}_2=\mathbf{p}_3=\mathbf{p}_4=(0, 0, 0, 0)$ that are obtained when $v_1$'s strategy is $s_1 = (v_2|v_3)$ and note that $\sw(\mathbf{P}) = 2$. If, however, $v_1$ selects strategy $s'_1=(v_3|v_2)$, we end up with the clearing payment $ \mathbf{p}'_1=(0, 0, 1, 0), \mathbf{p}'_2=(0, 0, 0, 0),\mathbf{p}'_3=(0, 0, 0, M)$, and $\mathbf{p}'_4=(0, 0, M, 0)$ and we obtain $\sw(\mathbf{P}') = 2M+2$. Hence, $\poa\geq \frac{2M+2}{2} = M+1$ which can become arbitrarily large.
\end{proof}

\section{Maximizing the equity}\label{sec:equities} 
In this section we consider the case of equities. Similarly, the social welfare is defined as the sum of equities. We present interesting properties of clearing \mk{payments} and observe that Nash equilibria always exist in such games, contrary to the case of total assets.

\subsection{Existence and properties of equilibria}

We warm up with a known statement in absence of default costs; the short proof is included here for completeness.
%\begin{remark}\label{rem:constant-equities}
%The sum of firm equities is always equal to the sum of external assets, even with CDS contracts, i.e., $\sw(\pp) = \sum_{i\in [n]}{e_i}$,  under any \mk{(not necessarily maximal)} clearing payments $\pp$.
%\end{remark}
%
%\begin{proof-folklore}  First, observe that for given clearing payments $\mathbf{P}$ and each firm $v_i$, it holds $E_i(\mathbf{P}) = \max\{0, a_i(\mathbf{P})-L_i\} = e_i+\sum_{j\in [n]}{p_{ji}}-\sum_{j\in [n]}{p_{ij}}$, as $\sum_{j\in [n]}{p_{ij}}=\min\{e_i+\sum_{j\in [n]}{p_{ji}}, L_i\}$ since $\mathbf{P}$ are clearing payments. By summing over all firms, we have $$\sw(\mathbf{P}) = \sum_{i\in [n]}{E_i(\mathbf{P})} = \sum_{i\in [n]}{\max\{0, a_i(\mathbf{P})-L_i\}} = \sum_{i\in [n]}{\left(e_i+\sum_{j\in [n]}{p_{ji}}-\sum_{j\in [n]}{p_{ij}}\right)} = \sum_{i\in [n]}{e_i},$$ since the total incoming payments across all firms equal the total outgoing payments.
%\end{proof-folklore}
%Informally, in absence of default costs, money can be just redistributed among firms. Let us remark, however, that this may no longer hold once we introduce default costs or negative external assets. 
In particular, each firm obtains the same equity under all clearing payments, so it does not have a preference; \mk{this provides additional justification to our assumption to limit our attention to maximal clearing payments computed in Section~\ref{sec:non-strategic}, in case of various clearing payments}.

\begin{lemma}[\cite{GSRB17}]\label{lem:same-equity}
Each firm obtains the same equity under different clearing payments, given a strategy profile. That is, given the firms' strategies, for any two different   \mk{(not necessarily maximal)} clearing payments $\mathbf{P}$ and $\mathbf{P}'$, it holds $E_i(\mathbf{P}) = E_i(\mathbf{P}')$ for each firm $v_i$.
\end{lemma}

\begin{proof}
Let $\mathbf{P}^*$ be the maximal clearing  payments and let $\mathbf{P}$ be any other clearing payments.
The corresponding equities for each firm $v_i$ are
\begin{align}\label{eq:opt-pp-clearing}
E_i(\mathbf{P}^*) = \max(0, e_i+\sum_{j\in [n]}{p^*_{ji}-L_i}) = e_i + \sum_{j\in [n]}{p^*_{ji}} - \sum_{j\in [n]}{p^*_{ij}}
\end{align}
and
\begin{align}\label{eq:pp-clearing}
E_i(\mathbf{P}) = \max(0, e_i+\sum_{j\in [n]}{p_{ji}-L_i}) = e_i + \sum_{j\in [n]}{p_{ji}} - \sum_{j\in [n]}{p_{ij}},
\end{align}
respectively. The rightmost equalities above hold, since each firm either pays all its liabilities, if it is solvent, or uses all its external and internal assets to pay part of its liabilities, if it is in default.

Using (\ref{eq:opt-pp-clearing}) and (\ref{eq:pp-clearing}) and by summing over all firms, we obtain
\begin{equation*}
\sum_{i\in [n]}{E_i(\mathbf{P}^*)}-\sum_{i\in [n]}{E_i(\mathbf{P})} = \sum_{i\in [n]}{\sum_{j\in [n]}{\left(p^*_{ji}-p^*_{ij}\right)}} - \sum_{i\in [n]}{\sum_{j\in [n]}{\left(p_{ji}-p_{ij}\right)}} =0,
\end{equation*}
as, for any clearing payments, the total incoming payments equal the total outgoing payments.

We remark that, since $\mathbf{P}^*$ are maximal clearing payments we get that $\sum_{j\in [n]}{p^*_{ij}}\geq \sum_{j\in [n]}{p_{ij}}$ for each firm $v_i$ and, hence, by (\ref{eq:opt-pp-clearing}) and (\ref{eq:pp-clearing}) it can only be that $E_i(\mathbf{P}^*)\geq E_i(\mathbf{P})$. Therefore, since we have shown that $\sum_{i\in [n]}{E_i(\mathbf{P}^*)}=\sum_{i\in [n]}{E_i(\mathbf{P})}$, we conclude that $E_i(\mathbf{P}^*)= E_i(\mathbf{P})$ for each firm~$v_i$.
\end{proof}

Lemma \ref{lem:same-equity} also indicates that, for any given strategy profile, any firm is always either solvent or in default in all resulting clearing payments. We exploit this property to obtain the following result; this extends a result by Papp and Wattenhofer (Theorem 7 in \cite{wattenhofer2020network}) which holds for the maximal clearing payments.

\begin{theorem}\label{thm:equity-all-clearing}
\mk{Even with CDS contracts, any strategy profile is a Nash equilibrium, when firms aim to maximize their equity. This holds even if the clearing payments that are realized are not maximal.}
\end{theorem}
\begin{proof}
\mk{Assume otherwise that there exists a strategy profile which is not a Nash equilibrium.} Let $v_i$ be a firm that wishes to deviate and $s_i$ be its strategy. Clearly, if $v_i$ is solvent it can repay all its liabilities in full and, hence, the payment priorities are irrelevant. So, let us assume that $v_i$ is in default. If, by deviating, $v_i$ remains in default, then its equity remains $0$. Therefore, assume that $v_i$ deviates to another strategy $s'_i$ where it is solvent. In that case, however, $v_i$ could fully repay its liabilities regardless of the payment priorities and, therefore, it can still repay its liabilities when playing $s_i$; a contradiction. 
\end{proof}

Note that Lemma \ref{lem:same-equity} no longer holds once default costs are introduced; see e.g., Example 3.3 in \cite{rogers2013failure} where both firms, each having a singleton strategy set, may be in default or solvent depending on the clearing payments. The next result extends Theorem 7 in \cite{wattenhofer2020network} to the setting with default costs, and guarantees the existence of Nash equilibria (and actually strong ones) when firms wish to maximize their equity.

\begin{theorem}\label{thm:equity-maximal-negative}
Even with default costs and negative external assets, any strategy profile is a strong equilibrium when firms aim to maximize their equity.
\end{theorem}

\begin{proof}
We begin by transforming an instance $\mathcal{I}$ with negative external assets into another instance ${\mathcal I}'$ without negative external assets, albeit with a slightly restricted strategy space for each firm. In particular, we add an auxiliary firm $t$ and define liabilities and assets as follows. For any pair of firms $v_i$, $v_j$ which does not include $t$, we set $l'^0_{ij}=l^0_{ij}$. For each firm $v_i$ with $e_i<0$, we set $e'_i=0$ and set liability $l'^0_{it} = e_i$, while for any other firm $i'$ we set $e'_{i'} = e_{i'}$ and $l'^0_{i't}=0$. Furthermore, we restrict the strategy space of each firm in ${\mathcal I}'$ so that their topmost priority class includes only firm $t$. An example of this process is shown in  Figure~\ref{fig:negative assets}.
%The proof then follows by Theorem \ref{thm:equity-maximal-clearing}. 
\begin{figure}[h]
\centering
\includegraphics[height=2cm]{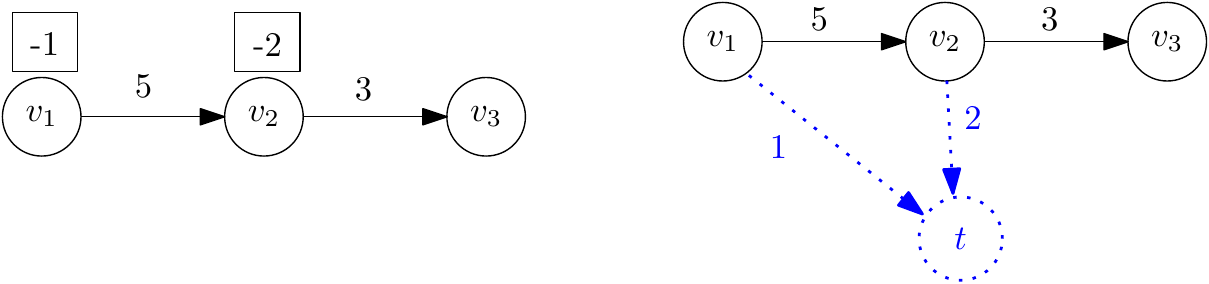}
\caption{Transforming an instance with negative external assets into an instance without negative external assets.}\label{fig:negative assets}
\end{figure}

Given a strategy profile for instance ${\mathcal I}$, we create the corresponding strategy profile for instance ${\mathcal I}'$ by having firm $t$ as the single topmost priority creditor and, then, append the initial strategy profile. It holds that, for any given strategy profile, the maximal clearing payments for instance ${\mathcal I}$ corresponds to  maximal clearing payments in ${\mathcal I}'$ for the new strategy profile; this can be easily proved by contradiction. We can now proceed with the proof by assuming non-negative external assets, without loss of generality. 

Consider maximal clearing payments $\mathbf{P}^*$  and the associated strategy profile $\mathbf{s}$. Let us assume that there is a coalition of firms $C=\{v_{C_1},\ldots,v_{C_k}\}$ where each member of the coalition can strictly increase its equity after a joint deviation. In particular, let $v_{C_i}$, for $i=1,\ldots,k$ change its strategy from $s_{C_i}$ to $s'_{C_i}$.  Clearly, each $v_{C_i}$ must have a strictly positive equity in the resulting new  maximal clearing payments $\mathbf{P}'$, since its equity was $0$ before. But then, each $v_{C_i}$  should remain solvent under strategy $s_{C_i}$ as well, and therefore $v_{C_i}$'s actual payment priorities are irrelevant, for $i=1,\ldots,k$. Note that $\mathbf{P}'$ should also be maximal clearing payments under the initial strategy profile; a contradiction to the maximality of $\mathbf{P}^*$.
\end{proof}

\subsection{(In)efficiency of equilibria}
 We start by noting that Lemma \ref{lem:same-equity} together with Theorem \ref{thm:equity-all-clearing} imply the following, which we note holds for any payment scheme.
\begin{corollary}
The price of anarchy in financial network games with CDS contracts is $1$ when firms aim to maximize their equity.
\end{corollary}

The above positive result, however, no longer holds when default costs or negative external assets exist. For these cases we derive the following results.

%We conclude this section by studying priority-proportional payments and the corresponding strategic games where firms can assign their creditors to priority classes. For this class of strategic games, we also present results on the price of anarchy and stability when firms aim to maximize their equity and total assets.

%In the remaining of this section we consider the case where firms wish to maximize their equity. 
\begin{theorem}\label{thm:equity-default-costs}
%The price of anarchy of priority-proportional games with default costs $(\alpha, \beta)$ where $\beta\in (0,1)$ is unbounded when firms aim to maximize their equity.
The price of anarchy with default costs $(\alpha, \beta)$ when firms aim to maximize their equity, is
\begin{enumerate}[(i)]
\item $1$ when $\alpha=\beta=0$,
\item unbounded when: i) $\beta\in (0,1)$, ii) $\beta=0$ and $\alpha\in (0,1]$, or iii) $\beta=1$ and $\alpha=0$,
\item at least $1/\alpha -\varepsilon$, if $\beta=1$ and $\alpha\in (0,1)$ for \mk{any} $\varepsilon>0$.
\end{enumerate}
\end{theorem}

\begin{proof}
We begin with the case $\alpha=\beta=0$. We claim that all strategy profiles correspond to the same clearing payments hence admit the same social welfare. 
%Indeed, consider two different strategy profiles $\mathbf{s}$ and $\mathbf{s}'$. Since $\alpha=\beta=0$ it holds that the payment of any firm $i$ in any strategy profile is either $0$, if it is in default, or $L_i$, otherwise. 
It suffices to observe that neither the strategy of a solvent firm, nor the strategy of a firm in default, affect the set of firms in default and consequently the clearing payments.  Consider two strategy profiles $\mathbf{s}$ and $\mathbf{t}$. A firm $i$ that is solvent under $\mathbf{s}$, will be solvent and continue to make payments $p_i=L_i$ under any strategy (given the strategies of everyone else), i.e., will be solvent at the clearing payments consistent with strategy vector $(t_i,\mathbf{s}_{-i})$, derived by $\mathbf{s}$ if firm $i$ alone changes her strategy from $s_i$  to $t_i$. On the other hand, a firm $i$ that is in default under $\mathbf{s}$, will similarly remain in default \mk{under} $(t_i,\mathbf{s}_{-i})$ and continue to make $0$ payments since $\alpha=\beta=0$, thus not affecting the set of firms in default. The claim follows by considering the individual deviations from $s_i$ to $t_i$ of all firms $i$ sequentially and observing that the set of firms in default and, hence, the clearing payments are unaffected at each step. 

\begin{figure}[htbp]
\centering
\includegraphics[width=7cm]{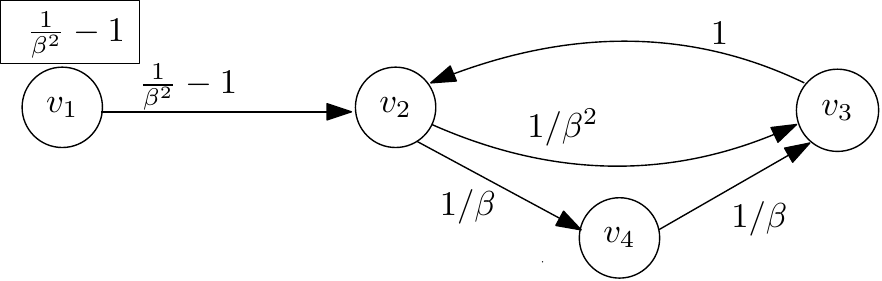}
\caption{The instance in the proof of Theorem \ref{thm:equity-default-costs}, where we assume $\beta \in (0,1)$.}
\label{fig:default cost}
\end{figure}
Regarding the case $\beta \in (0,1)$, consider the financial network in Figure \ref{fig:default cost}. Clearly, only firm $v_2$ can strategize and observe that it is always in default irrespective of its strategy. When $v_2$ selects strategy $s_2 = (v_4|v_3)$, we obtain the clearing payments $\mathbf{p}_1=(0, \frac{1}{\beta^2}-1, 0, 0), \mathbf{p}_2=(0, 0, 0, \frac{1+\beta}{\beta(\beta^2+\beta+1)}), \mathbf{p}_3=(0, \frac{\beta(1+\beta)}{\beta^2+\beta+1},  0, 0)$, and $\mathbf{p}_4=(0, 0, \frac{1+\beta}{\beta^2+\beta+1},0)$.  Notice that $v_2$, $v_3$ and $v_4$ are in default and we have $\sw(\pp)=0$.

When, however, $v_2$ selects strategy $s'_2 = (v_3|v_4)$, we obtain the clearing payments $ \mathbf{p}'_1=(0, \frac{1}{\beta^2}-1, 0, 0), \mathbf{p}'_2=(0, 0, \frac{1}{\beta}, 0), \mathbf{p}'_3=(0, 1, 0, 0)$, and $\mathbf{p}'_4=(0, 0, 0, 0)$. Now, $v_2$ and $v_4$ are in default and we have $\sw(\pp') = 1/\beta-1$. Since $\opt\geq \sw(\pp')$, the claim follows.

Regarding the case $\beta=0$ and $\alpha\in (0,1]$, consider the financial network in Figure \ref{fig:Thm10} and note that $v_2$ is always in default irrespective of its strategy. If $v_2$ selects strategy $s_2 = (v_6|v_3)$, then $v_6$ ends up having equity $\alpha$. However, strategy $(v_3|v_6)$ is also an equilibrium strategy for $v_2$, but it results in each firm having equity $0$.

\begin{figure}[h]
	\centering
	\includegraphics[width=8cm]{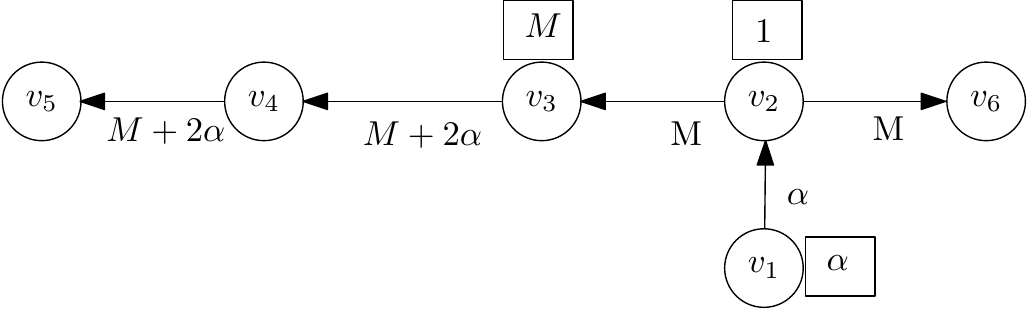}
	\caption{The instance in the proof of Theorem \ref{thm:equity-default-costs}, when $\beta \in \{0,1\} $ and $\alpha\neq \beta$.}
	\label{fig:Thm10}
\end{figure}

Regarding the remaining case, i.e., $\beta=1$ and $\alpha\in [0,1)$, consider again the financial network in Figure \ref{fig:Thm10}. As before, $v_2$ is always in default. If $v_2$ selects strategy $(v_3|v_6)$, then $v_5$ ends up having equity $M+2\alpha$. However, $(v_6|v_3)$ is also an equilibrium strategy for $v_2$, but it results in equity $2\alpha$ for $v_6$, equity $\alpha M$ for $v_5$ and equity $0$ for the remaining firms. The theorem follows by straightforward calculations.
\end{proof}

We complement this result with a tight upper bound for the case $\beta=1$. %The proof follows along similar lines to Remark \ref{rem:constant-equities}.
\begin{theorem}\label{thm:poa-upper-beta-1}
The price of anarchy with default costs $(\alpha, 1)$ when firms aim to maximize their equity is at most $1/\alpha$. \mk{This holds even if the clearing
payments that are realized are not maximal.}
\end{theorem} 
\begin{proof}
Consider any clearing payments $\mathbf{P}$ and let $S(\mathbf{P})$ and $D(\mathbf{P})$ be the set of solvent and in default firms under $\mathbf{P}$. Recall that for any firm $v_i$, we have $E_i(\mathbf{P}) = \max\{0,a_i(\pp)-L_i\}$. For any firm $v_i \in S(\mathbf{P})$, it holds that $E_i(\mathbf{P}) = e_i+\sum_{j\in[n]}{p_{ji}} - \sum_{j\in[n]}{p_{ij}}$, as $\sum_{j\in [n]}{p_{ij}}=\min\{e_i+\sum_{j\in [n]}{p_{ji}}, L_i\}$. Similarly, for any $v_i \in D(\mathbf{P})$, we have $E_i(\mathbf{P}) = \alpha e_i + \sum_{j\in[n]}{p_{ji}} - \sum_{j \in [n]}{p_{ij}}$, since $\beta = 1$. By summing over all firms, we get
\begin{align*}
E(\mathbf{P}) &= \sum_{i: v_i \in S(\mathbf{P})}{E_i(\mathbf{P})} + \sum_{i: v_i \in D(\mathbf{P})}{E_i(\mathbf{P})}\\
&= \sum_{i\in[n]}{\left(e_i +\sum_{j\in [n]}{p_{ji}} - \sum_{j\in [n]}{p_{ij}}\right)} - (1-\alpha)\sum_{i: v_i \in D(\mathbf{P})}{e_i}\\
&= \sum_{i\in[n]}{e_i} - (1-\alpha)\sum_{i: v_i \in D(\mathbf{P})}{e_i}.
\end{align*} 
Clearly, the social welfare is maximized when $D(\mathbf{P}) = \emptyset$, i.e., $OPT\leq \sum_{i\in [n]}{e_i}$, while it is minimized when $D(\mathbf{P})$ includes all firms, i.e., $E(\mathbf{P}) \geq \alpha\sum_{i \in [n]}{e_i}$; this completes the proof.
\end{proof}

Extending the setting to allow for negative external assets again may lead to unbounded social welfare loss.
\begin{theorem}\label{thm:equity-negative}
The price of anarchy is unbounded when negative external assets are allowed and firms aim to maximize their equity.
\end{theorem}

\begin{proof}
Consider the financial network in Figure \ref{fig:negative}. Clearly, only $v_1$ can strategize but, irrespective of its strategy, $v_1$ is always in default and $E_1=0$; hence, any strategy profile is a Nash equilibrium. When $v_1$ selects strategy $s_1 = (v_2|v_3)$ we obtain the clearing payments $\mathbf{p}_1=(0, 1, 0), \mathbf{p}_2=(0, 0, 0), \mathbf{p}_3=(0, 0, 0)$. However, when $v_1$ selects strategy $s'_1 = (v_3|v_2)$ we obtain the clearing payments $\mathbf{p}'_1=(0, 0, 1), \mathbf{p}'_2=(0, 0, 0), \mathbf{p}'_3=(0, 0, 0)$ 
%\footnote{This reveals one shortcoming of expressing social welfare in terms of payment vectors instead of payment matrices, as identical payment vectors may lead to different social welfare. Note that this cannot occur when the social welfare is defined as the total assets.} 
with $\sw(\mathbf{P}') = 1$. Since $\opt\geq \sw(\mathbf{P}')$, the claim follows.
\begin{figure}[htbp]
\centering
\includegraphics[width=4cm]{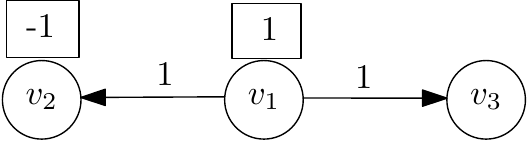}
\caption{A financial network with negative external assets and unbounded price of anarchy.}
\label{fig:negative}
\end{figure}
\end{proof}

Still, in the presence of default costs or negative external assets, Theorem \ref{thm:equity-maximal-negative} leads to the next positive result.
\begin{corollary}
The strong price of stability is $1$ even with default costs  and negative external assets, when firms aim to maximize their equity.
\end{corollary}

%\todo[inline]{CHECK}
%One can extend this result also to the case of strong equilibria and the strong price of stability. In a strong equilibrium \cite{Aum59}, no joint deviation of any coalition of firms leads to strictly greater utility for all deviating firms. Indeed, any solvent firm already pays all its liabilities in full, and, ..... Similarly, any firm in default has equity $0$ and ....

%\begin{corollary}
%The strong price of stability of priority-proportional games with default costs is always $1$ even with negative external assets, when firms aim to maximize their equity.
%\end{corollary}

%\todo[inline]{END OF check block}
If we relax the stability notion and consider the price of stability in super-strong equilibria, where a coalition of firms deviates if at least one strictly improves its utility and no firm suffers a decrease in utility, then we obtain a negative result.

\begin{theorem}\label{thm:superstrong}
The price of stability in super-strong equilibria is unbounded when negative external assets are allowed and firms aim to maximize their equity.
\end{theorem}

\begin{proof}
\begin{figure}[h]
\centering
\includegraphics[height=2.8cm]{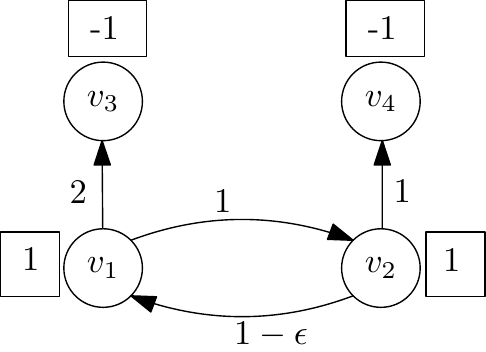}
\caption{A financial network with negative external assets and unbounded price of stability in super-strong equilibria.}
\label{fig:strong PoS}
\end{figure}

Consider the financial network as shown in Figure \ref{fig:strong PoS}, where $\varepsilon>0$ is a small constant. Observe that $v_1$ is always in default since even its maximal possible total asset $a_1=2-\varepsilon$ is still less than its total liabilities $L_1=3$, which means $E_1=0$ in any scenario. If $v_1$'s strategy is other than $s_1 = (v_2|v_3)$, then $v_2$ is always in default no matter what strategies it chooses, that is $E_2=0$. Thus, any strategy profile where $v_1$ does not prioritize the payment of $v_2$ is a Nash equilibrium with $E_1=E_2=0$. However, if $v_1$ and $v_2$ form a coalition, then the only superstrong equilibrium occurs when $v_1$ and $v_2$ prioritize the payment of each other, resulting in the clearing payments $\mathbf{p}_1=(0, 1, 1-\epsilon,0), \mathbf{p}_2=(1-\epsilon, 0, 0, 1), \mathbf{p}_3=\mathbf{p}_4=(0, 0, 0,0)$ with $\sw(\pp)=\varepsilon$. Furthermore, when $v_1$ follows the strategy $s'_1 = (v_3|v_2)$ and $v_2$ follows strategy $s'_2 = (v_1|v_4)$, we obtain the clearing payments $\mathbf{p}'_1=(0, 0, 2-\epsilon, 0), \mathbf{p}'_2=(1-\epsilon, 0, 0, \epsilon), \mathbf{p}'_3=\mathbf{p}'_4=(0, 0, 0,0)$ with $\sw(\mathbf{P}')=1-\varepsilon$. Since $\opt\geq \sw(\mathbf{P}')$, we obtain that the superstrong price of stability is at least $\frac{1-\varepsilon}{\varepsilon}$ and the theorem follows since $\varepsilon$ can be arbitrarily small.
\end{proof}

\section{Conclusions and discussion}\label{sec:conclusions}
We have studied strategic payment games in financial networks with priority-proportional payments and we have presented an almost full picture both with respect to structural properties of  clearing payments as well as the quality of Nash equilibria.

Our work reveals several interesting open questions. In particular, an interesting restriction on the strategy space is to always prioritize strategies that could lead to incoming payments, i.e., when each firm $i$ always ranks higher a payment that might create a cash flow through a directed cycle (in the liabilities graph) back to $i$, than a payment that does not. While some of our negative results would still hold, e.g., Theorem \ref{thm:pp-no-nash-equilibria}, others, such as Theorem \ref{thm:pp-pos-positivedefault}, crucially rely on the reverse ranking. Furthermore, one could consider addressing computational complexity questions like deciding whether a Nash equilibrium exists or computing a Nash equilibrium when it is guaranteed to exist, when firms aim to maximize their total assets.

\bibliographystyle{splncs04}
\bibliography{ref}

\newpage
\appendix 
%\section{Proofs of Section \ref{sec:non-strategic}}
%\input{appendix_non-strategic}

\section{Proof of Theorem \ref{thm:pp-no-nash-equilibria} (cont'd)}\label{app:no-nash-equilibria}

\begin{table}[h!]
\centering
  \begin{tabular}{| c || c | c | c |}
    \hline
      \multicolumn{4}{| c |}{$s_1 = (v_6|v_7)$} \\\hline
    & $s_3 = (v_1|v_5)$  & $s_3 = (v_1, v_5)$  & $s_3 = (v_5|v_1)$  \\
    \hline
    $s_2 = (v_1|v_4)$ & $9, 4.5, 4.5$ & $\frac{4.5-\varepsilon}{1-\varepsilon}, 4.5, \frac{3.5}{1-\varepsilon}$ & $4.5, 4.5, 3.5$ \\
    \hline
    $s_2 = (v_1, v_4)$  & $2+\frac{5\varepsilon}{1-\varepsilon}, \frac{5}{1-\varepsilon}, 2$ & $\frac{6\varepsilon}{1-\varepsilon}, \frac{3+3\varepsilon}{1-\varepsilon}, 3$ & $\frac{3\varepsilon}{1-\varepsilon}, \frac{3}{1-\varepsilon}, 3$ \\
    \hline
    $s_2 = (v_4|v_1)$  & $2, 5, 2$  & $3\varepsilon, 3+3\varepsilon, 3$  & $0, 3, 3$  \\\hline\hline
      \multicolumn{4}{| c |}{$s_1 = (v_7|v_6)$} \\\hline
    & $s_3 = (v_1|v_5)$  & $s_3 = (v_1, v_5)$  & $s_3 = (v_5|v_1)$  \\
    \hline
    $s_2 = (v_1|v_4)$ & $9, 4.5, 4.5$ & $2+\frac{5\varepsilon}{1-\varepsilon}, 2, \frac{5}{1-\varepsilon}$ & $2, 2, 5$   \\
    \hline
    $s_2 = (v_1, v_4)$  & $\frac{4.5-\varepsilon}{1-\varepsilon}, \frac{3.5}{1-\varepsilon}, 4.5$ &  $\frac{6\varepsilon}{1-\varepsilon}, 3, \frac{3+3\varepsilon}{1-\varepsilon}$ & $3\varepsilon, 3, 3+3\varepsilon$ \\
    \hline
    $s_2 = (v_4|v_1)$  &  $4.5, 3.5, 4.5$ & $\frac{3\varepsilon}{1-\varepsilon}, 3, \frac{3}{1-\varepsilon}$ & 0, 3, 3 \\\hline\hline
      \multicolumn{4}{| c |}{$s_1 = (v_6, v_7)$} \\\hline
    & $s_3 = (v_1|v_5)$  & $s_3 = (v_1, v_5)$  & $s_3 = (v_5|v_1)$  \\
    \hline
    $s_2 = (v_1|v_4)$ &  $9, 4.5, 4.5$ & $\frac{4+6\varepsilon}{1-\varepsilon}, \frac{4+\varepsilon}{1-\varepsilon}, \frac{5}{1-\varepsilon}$ & $4, 4, 5$ \\
    \hline
    $s_2 = (v_1, v_4)$  & $\frac{4+6\varepsilon}{1-\varepsilon}, \frac{5}{1-\varepsilon}, \frac{4+\varepsilon}{1-\varepsilon}$ & $\frac{6\varepsilon}{1-\varepsilon}, \frac{3}{1-\varepsilon}, \frac{3}{1-\varepsilon}$ & $\frac{6\varepsilon}{2-\varepsilon}, \frac{6}{2-\varepsilon}, \frac{6\varepsilon}{2-\varepsilon}$ \\
    \hline
    $s_2 = (v_4|v_1)$  & $4, 5, 4$ &  $\frac{6\varepsilon}{2-\varepsilon}, \frac{6}{2-\varepsilon}, \frac{6}{2-\varepsilon}$ & $0, 3, 3$  \\\hline
  \end{tabular}
\caption{The table of utilities for the network without Nash equilibria; note that $\varepsilon = \frac{6}{M+6}$. Each cell entry contains the utilities of $v_1$, $v_2$, $v_3$ in that order. Each $3\times 3$ subtable corresponds to a fixed strategy for $v_1$.}\label{tab:inexistence}
\end{table}

%\input{appendix_assets}

%\section{Proofs of Section \ref{sec:equities}}
%\input{appendix_equities}
%\input{appendixproof.tex}
\end{document}